%% file: realSUSYMathPhys.tex
\ifx\pdfoutput\undefined\else\pdfoutput=1\fi
\documentclass[9pt,a4paper,oneside,reqno]{amsart} 
\input{packages.tex}

\input{notations.tex}
\input{bibliography_setup.tex}
\bibliography{references}

\newtheorem{theorem}{Theorem}
\newtheorem{lemma}[theorem]{Lemma}
\newtheorem{remark}[theorem]{Remark}
\newtheorem{corollary}[theorem]{Corollary}
\numberwithin{theorem}{section}

\newcommand{\f}{\mathfrak{f}}
\newcommand{\g}{\mathfrak{g}}

\date{\today}
\author{Giorgio Cipolloni \and L\'aszl\'o Erd\H{o}s}
\address{IST Austria, Am Campus 1, 3400 Klosterneuburg, Austria}
\author{Dominik Schr\"oder\(^{\ast}\)}
\address{Institute for Theoretical Studies, ETH Zurich, Clausiusstr.\ 47, 8092 Zurich, Switzerland}
\email{giorgio.cipolloni@ist.ac.at} 
\email{lerdos@ist.ac.at}
\email{dschroeder@ethz.ch}
\thanks{\(^\ast\)Supported by Dr.\ Max R\"ossler, the Walter Haefner Foundation and the ETH Z\"urich Foundation}
\subjclass[2010]{60B20, 15B52, 68W40} 
\keywords{Supersymmetric formalism, Circular Law, Superbosonization formula}
\title[Density of small singular values of the shifted real Ginibre ensemble]{Density of small singular values of the shifted real Ginibre ensemble}
\date{\today}
\begin{document}   
\thispagestyle{empty}  

\begin{abstract}  
    We derive a precise asymptotic formula  for the density of the small singular values
    of the real Ginibre matrix ensemble shifted by a complex parameter $z$
    as the  dimension tends to infinity.
    For $z$ away from the real axis the formula coincides with that for the
    complex Ginibre ensemble we derived earlier in~\cite{Cipolloni2020}. 
    On the level of the one-point function of the low lying singular values we thus confirm the
    transition from  real to complex Ginibre ensembles
    as the shift parameter $z$ becomes genuinely complex; the analogous 
    phenomenon has been well known for eigenvalues. 
    We use the superbosonization formula~\cite{MR2430637} in a regime where 
    the main contribution comes from a three dimensional  saddle manifold.
\end{abstract}

\maketitle

\section{Introduction}

The universality paradigm in random matrix theory asserts that the
local eigenvalue statistics of large random matrices  depend only
on the basic symmetry class of the ensemble. 
In the Hermitian case, this dependence is usually  investigated for the $k$-point
functions  starting from  $k\ge 2$, while the one-point function is largely insensitive to the 
symmetry class apart from finite-size correction terms (see, e.g.~\cite{MR3137043} for GOE/GUE). 
For non-Hermitian matrices, however, the real axis plays an interesting distinguishing 
role between the real and complex ensembles already on the level of the one-point function.
In the simplest Gaussian case this phenomenon has been well known for the eigenvalues; in this paper we 
investigate it for singular values where no explicit formulas are available.

We consider the \emph{real} or \emph{complex Ginibre ensemble}~\cite{MR173726}, i.e.\ 
large $N\times N$ random matrices $X$ with independent identically distributed (i.i.d) 
real or complex Gaussian
entries $x_{ab}$.  The customary normalization, $\E x_{ab}=0$, $\E |x_{ab}|^2= N^{-1}$,  guarantees
that the density of eigenvalues of $X$ converges to the uniform measure on the complex unit
disk $\set[\big]{ z\given \abs{z}\le 1}$, known as the \emph{circular law}, and that the spectral radius  of $X$
converges to 1 with very  high  probability (these results also hold for general non-Gaussian matrix elements, see
e.g.~\cite{MR773436,MR1428519,MR2409368,MR866352,MR863545,MR3813992,2012.05602}). 

While the distribution of the complex Ginibre eigenvalues is clearly rotationally invariant,
the real axis plays a special role for the real Ginibre ensemble, in particular 
there are typically  $\sim \sqrt{N}$ real eigenvalues~\cite{MR1231689} (see also the exact formula for having precisely  $k$ real eigenvalues
in~\cite{MR2185860}). 
In fact, all correlation functions
of the Ginibre eigenvalues are explicitly known see~\cite{MR173726} and~\cite{MR0220494} for  the simpler complex case, and~\cite{MR1121461,MR1437734,MR2530159,17930739} for the more involved real case.
The precise formulas  reveal a remarkable phenomenon~\cite[Theorem 11]{MR2530159}: the local  \emph{eigenvalue} statistics for real Ginibre matrices coincide with those for complex Ginibre matrices anywhere in the spectrum away from the real axis (see also~\cite{MR4047447}).

To what extent does this phenomenon hold for 
low lying \emph{singular values} 
of $X$ and their shifted version $X-z$ with a complex parameter $z$? 
While singular values may behave very differently than eigenvalues, intuitively
the very small singular values of $X-z$ are still  related to the eigenvalues of $X$ near $z$, 
since $z$ is an eigenvalue of $X$ if and only if $X-z$ has a zero singular value.
Hence we expect that these small singular values  of $X-z$
for $z$ away  from the real axis behave in the same way for real and complex Ginibre matrices.   
 This was recently proven in~\cite[Theorem 2.8]{MR4235475} for all $k$-point correlation functions and even for 
any i.i.d.\ (i.e.\ not necessarily Gaussian)
distributed matrix elements but only in the regime $|\Im z|\sim 1$. 
In this paper  we
prove that this phenomenon holds down to very close to the imaginary axis, $|\Im z|\gg N^{-1/2}$, 
on the level of the density (or one-point
function) of the singular values  using supersymmetric (SUSY) techniques. 
In our related paper~\cite{2105.13719} we explore  the power of  this  approach 
with applications to numerical analysis
by establishing new bounds on the eigenvector condition number
and on the eigenvector overlaps~\cite{PhysRevLett.81.3367,MR3851824,MR4095019}.

More precisely, we find that in the large $N$ limit
the density of the low lying singular values of $X-z$ for a real Ginibre matrix
coincides  with that of the complex Ginibre matrix $X$
as long as $|\Im z|\gg N^{-1/2}$, while it is different for $|\Im z|\sim N^{-1/2}$, c.f.\ Figure~\ref{hist figure}. 
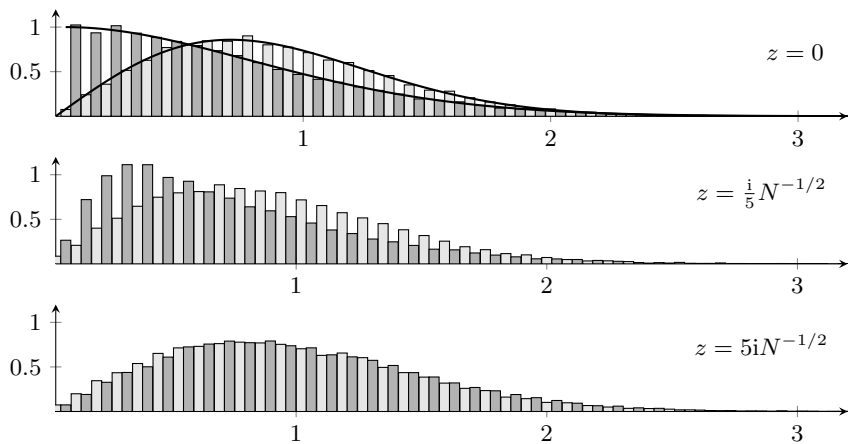
\begin{figure}[htbp]
    \centering
    \begin{tikzpicture}
        \begin{axis}[width=12cm,height=3cm,
            ymin=0,xmax=3.2,ymax=1.2, ytick={0,0.5,1}, xtick={0,1,2,3},
            axis lines=center, 
            ]
            \addplot[ybar,bar width=.04,fill=black!10,bar shift=0.0] table[col sep=comma,x index=0,y index=1] {hists.csv};
            \addplot[ybar,bar width=.04,fill=black!30,bar shift=0.04] table[col sep=comma,x index=0,y index=2] {hists.csv} node[above,yshift=2em,xshift=-10em] {$z=0$};
            \addplot [thick, domain=0:3.2, samples=201] {2*x*exp(-x^2)};
            \addplot [thick, domain=0.04:3.2, samples=201] {(1+x-.04)*exp(-(x-.04)^2/2-(x-.04))};
        \end{axis}
    \end{tikzpicture}
    \begin{tikzpicture} 
        \begin{axis}[width=12cm,height=3cm,
            ymin=0,xmax=3.2,ymax=1.2, ytick={0,0.5,1}, xtick={0,1,2,3},
            axis lines=center, 
            ]
            \addplot[ybar,bar width=.04,fill=black!10,bar shift=0.0] table[col sep=comma,x index=0,y index=3] {hists.csv};
            \addplot[ybar,bar width=.04,fill=black!30,bar shift=0.04] table[col sep=comma,x index=0,y index=4] {hists.csv} node[above,yshift=2em,xshift=-11.5em] {$z=\frac{\ii}{5}N^{-1/2}$}; 
        \end{axis}
    \end{tikzpicture}
    \begin{tikzpicture}
        \begin{axis}[width=12cm,height=3cm,
            ymin=0,xmax=3.2,ymax=1.2, ytick={0,0.5,1}, xtick={0,1,2,3},
            axis lines=center, 
            ]
            \addplot[ybar,bar width=.04,fill=black!10,bar shift=0.0] table[col sep=comma,x index=0,y index=5] {hists.csv};
            \addplot[ybar,bar width=.04,fill=black!30,bar shift=0.04] table[col sep=comma,x index=0,y index=6] {hists.csv} node[above,yshift=2em,xshift=-11.5em] {$z=5\ii N^{-1/2}$};
        \end{axis}
    \end{tikzpicture}
    \caption{Histogram of rescaled smallest singular value \(N\sigma_{\min }(X-z)\)  of \(X-z\) in the real (dark grey) and complex (light grey case). For \(z=0\) the asymptotic densities \(2xe^{-x^2}\) and \((1+x)e^{-x^2/2-x}\) have been computed by Edelman~\cite{MR964668}.}\label{hist figure} 
\end{figure}
This  indicates  a  transition in the  local singular value statistics of $X-z$ 
from real to complex as $|\Im z|$ increases beyond $N^{-1/2}$, similarly\footnote{We remark that~\cite[Theorem 11]{MR2530159} did not explicitly state that the transition takes place for $|\Im z|\gg N^{-1/2}$,
but it can be concluded from its proof.}
to the local eigenvalue
statistics  of $X$.   

Technically,
we express the averaged trace
of the resolvent of $(X-z)(X-z)^*$
in terms of contour integrals using the superbosonization 
formula~\cite{MR2430637} and perform the large $N$ limit. This  analysis
has been carried out for the complex case in~\cite{Cipolloni2020}, now we handle 
the considerably  more involved  real case. The main additional complication stems from 
the structure of the superbosonization formula: the  contour integration 
in the real case involves three integration variables, two of them are highly convolved
and their contours cannot be deformed independently; while 
the complex case has only two variables and the phase function is decoupled in them. 
The entire analysis is done at the bottom of the spectrum of $(X-z)(X-z)^*$, at a distance comparable with the 
(square of the) local spacing of the  singular values, hence our result directly gives precise information on 
individual singular values. 
In this critical regime the answer does not come simply  from a saddle point, but from a genuine three-fold
integral even after the $N\to\infty$ limit is taken.  With a careful choice of the interdependent
deformations of the contours we achieve the negative sign in the real part of the phase function
hence we can rigorously estimate the physically irrelevant highly oscillatory integration regimes. Note that the mere  existence of such deformation is not guaranteed by any physical principle,  let alone finding them explicitly
-- this is what we achieve here.
A further feature of our work is that we can 
handle  the bulk, $|z|< 1$, as well as
the edge regime, $|z|\approx 1$, where the scaling changes from $N^{-1}$ to $N^{-3/4}$.

\subsection*{Notations and conventions} 
For positive quantities \(f,g\) we write \(f\lesssim g\) and \(f\sim g\) if \(f \le C g\) or \(c g\le f\le Cg\), respectively, for some constants \(c,C>0\) which are independent of the basic parameters of the problem $N,\lambda,\widetilde{\eta},\widetilde{\delta}$ in~\eqref{eq:resvaretalam}. 
For any two positive, possibly $N$-dependent, quantities \(f,g\) we write \(f\ll g\) 
to denote that $f\lesssim N^{-\epsilon} g$, for some small $\epsilon>0$
(however this convention will be locally altered within the proof of Lemma~\ref{lem:smallatau}).
We abbreviate the minimum and maximum of real numbers by \(a\wedge b:=\min\set{a,b}\) 
and \(a\vee b:=\max\set{a,b}\).  

\section{Main results}
We consider the ensemble $Y^z:=(X-z)(X-z)^*$ with $X\in  \R ^{N\times N}$ being a \emph{real Ginibre} matrix, i.e.\ its entries $x_{ab}$ are such that $\sqrt{N}x_{ab}$ are i.i.d.\ real standard Gaussian random variables, and $z\in\C $ is a fixed complex parameter such that $|z|\le 1$. 
We compute the large $N$ asymptotics for the spectral one-point function 
$\E \Tr (Y^z-w)^{-1}$, with $w=E+\ii 0$.
The energy $E$ is chosen to be comparable with the local eigenvalue spacing of $Y^z$,
i.e.\ we study the small eigenvalues of $Y^z$. 
The imaginary part of $\E \Tr (Y^z-w)^{-1}$ is the density of states at the energy $E$. In particular, we focus on the transitional regime $|\Im z|\sim N^{-1/2}$ proving that $\E \Tr (Y^z-w)^{-1}$ exhibits a 
one-parameter family of behaviours 
depending on $N^{1/2}|\Im z|$. Additionally, we prove that  $\E \Tr (Y^z-w)^{-1}$ behaves as in the case of complex Ginibre matrix $X$ for $|\Im z|\gg N^{-1/2}$. 

In order to study the transitional regime $|\Im z|\sim N^{-1/2}$, we introduce the rescaled variables
\begin{equation}
    \label{eq:resvaretalam}
    \lambda:= N^{3/2}(1\vee \widetilde{\delta})E, \qquad \widetilde{\eta}:= N^{1/2}\Im z, \qquad \widetilde{\delta}:=N^{1/2}\delta,
\end{equation}
with $\delta:=1-|z|^2$. By~\cite[Section 5]{1907.13631} it is easy to see  that the level spacing of the eigenvalues 
of $Y^z$ close to zero  is of order 
\[
c(N,\widetilde{\delta}):=N^{-3/2}\cdot (1\wedge \widetilde{\delta}^{-1}),
\]
i.e., 
 for $|z|<1$ is given by $N^{-2}\delta^{-1}$ and for $|z|=1$ by $N^{-3/2}$,
which explains the scaling of $\lambda$. The unusual  $N^{-3/2}$ scaling 
in the edge regime $|z|=1$ originates from the fact that
the density of eigenvalues of the  Hermitized matrix
\begin{equation}\label{eq:herm}
    H^z: = \begin{pmatrix}  0 &  X-z \cr  (X-z)^* & 0 \end{pmatrix}
\end{equation}
features  a cubic cusp singularity that has a natural  eigenvalue spacing $N^{-3/4}$.

We now state the main  technical result  on the large $N$ asymptotics   of the one-point function. 
The main conclusion of the paper will be given as its Corollary~\ref{theo:1realcompl} afterwards. Note that the formulas~\eqref{eq:explt} are considerably simplified when $\widetilde\delta=0$, i.e.\ 
$|z|=1$, in particular, the spectral scaling factor becomes $c(N, \widetilde\delta=0)=N^{-3/2}$.
\begin{theorem}\label{theo:1pointreal}
    Let $C_0,C_1>0$ sufficiently large constants. For any $C_0^{-1}\le\lambda \le C_0$, for $\widetilde{\eta}=0$ or $C_0^{-1}\le|\widetilde{\eta}|\le C_0$, and for $\widetilde{\delta}=0$ or $\widetilde{\delta}\ge C_1$ it holds
    \begin{equation}
        \label{eq:largenas}
        \E  \frac{1}{N}\mathrm{Tr}(Y^z-\lambda c(N,\widetilde{\delta})-\ii 0)^{-1}=N^{1/2} I^{(\R )}(\lambda,\widetilde{\eta},\widetilde{\delta})+\mathcal{O}\left(1+\widetilde{\delta}\right),
    \end{equation}
    where
    \begin{equation}
        \label{eq:exp1poi}
        I^{(\R )}(\lambda,\widetilde{\eta},\widetilde{\delta}):=\frac{1}{4\pi \ii} \oint_{\Gamma} \dif \xi \int_\Omega \dif \tau  \int_{\Lambda} \dif a\frac{\xi^2a}{\tau^{1/2}} e^{\f(\xi,\lambda,\widetilde{\delta})-\g(a,\tau,\widetilde{\eta},\lambda,\widetilde{\delta})} G(a,\tau,\xi,\widetilde{\eta},\widetilde{\delta}),
    \end{equation}
    with $\Gamma$ any contour around $0$ in a counter-clockwise direction, $\Lambda$ any contour going out from $0$ in the direction of \(e^{\ii\pi/6}\) for a while  and then
     going to infinity in the direction $e^{3\pi\ii/5}$, and $\Omega$ any contour in the fourth quadrant going out from zero in the direction $e^{-\ii\pi/3}$ and ending in one with an angle $e^{\ii\pi/3}$, see Figure~\ref{contours}. Here
    \begin{equation}
        \label{eq:explt}
        \begin{split}
            \f(\xi,\lambda,\widetilde{\delta})&:= -(1\wedge \widetilde\delta^{-1})\lambda\xi+\frac{1}{2\xi^2}+\frac{\widetilde{\delta}}{\xi}, %
             \\
            \g(a,\tau,\widetilde{\eta},\lambda,\widetilde{\delta})&:= -(1\wedge \widetilde\delta^{-1})\lambda a+\frac{2\widetilde{\eta}^2(1-\tau)}{\tau}+\frac{2-\tau}{2a^2\tau^2}+\frac{\widetilde{\delta}}{a\tau}, \\
            G(a,\tau,\xi,\widetilde{\eta},\widetilde{\delta})&:= \frac{1}{a^2\tau\xi^6}+\frac{2}{a^3\tau^2\xi^5}+\frac{4-\tau}{a^4\tau^3\xi^4}+\frac{2}{a^5\tau^3\xi^3}+\frac{1}{a^6\tau^3\xi^2}+\frac{1}{a^2\tau\xi^4}+\frac{2}{a^3\tau^2\xi^3} \\
            &\quad+\frac{1}{a^4\tau^2\xi^2} +\frac{2\widetilde{\delta}}{a^2\tau\xi^5}+\frac{4\widetilde{\delta}}{a^3\tau^2\xi^4}+\frac{4\widetilde{\delta}}{a^4\tau^3\xi^3}+\frac{2\widetilde{\delta}}{a^5\tau^3\xi^2}+\frac{4\widetilde{\eta}^2}{a^2\tau^2\xi^4}+\frac{4\widetilde{\eta}^2}{a^3\tau^2\xi^3} \\
            &\quad+\frac{4\widetilde{\eta}^2}{a^4\tau^3\xi^2}+\frac{4\widetilde{\eta}^2\widetilde{\delta}}{a^2\tau^2\xi^3}+\frac{4\widetilde{\eta}^2\widetilde{\delta}}{a^3\tau^2\xi^2}.
        \end{split}
    \end{equation}
    The implicit constant in $\mathcal{O}(\cdot)$ depends on $C_0$. Moreover, the integral $I^{(\R )}(\lambda,\widetilde{\eta},\widetilde{\delta})$ is absolutely convergent and is bounded by
    $C(1+\widetilde{\delta})$ with
    a constant that depends only on $C_0$ and $C_1$.
\end{theorem}
\begin{figure}[htbp]
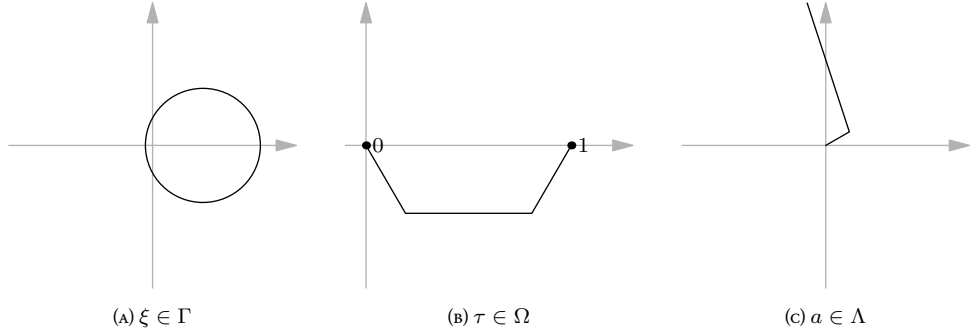

    \centering
    \begin{subfigure}[b]{0.3\textwidth}
        \centering
        \begin{asy}
            draw((-10,0)--(10,0),gray(.7),EndArrow);
            draw((0,-10)--(0,10),gray(.7),EndArrow);
            draw(circle((3.5,0),4));
        \end{asy}

        \caption{$\xi\in\Gamma$}\label{fig:y equals x}
    \end{subfigure}\hfill
    \begin{subfigure}[b]{0.3\textwidth}
        \centering
        \begin{asy}
            draw((-.1,0)--(1.3,0),gray(.7),EndArrow);
            draw((0,-.7)--(0,.7),gray(.7),EndArrow);
            draw((0,0)--(1/sqrt(27),-1/3)--(1-1/sqrt(27),-1/3)--(1,0));
            dot("$0$",(0,0));
            dot("$1$",(1,0));
        \end{asy}

        \caption{$\tau\in\Omega$}\label{fig:y equals x}
    \end{subfigure}\hfill
    \begin{subfigure}[b]{0.3\textwidth}
        \centering
        \begin{asy}
            draw((-4,0)--(4,0),gray(.7),EndArrow);
            draw((0,-4)--(0,4),gray(.7),EndArrow);
            draw((0,0)--aC(1,1)--aC(1,6));
        \end{asy}

        \caption{$a\in\Lambda$}\label{fig:y equals x}
    \end{subfigure}
    \caption{Depiction of the chosen contours}\label{contours}
\end{figure}
In Corollary~\ref{theo:1realcompl} below 
we study the behaviour of $I^{(\R )}(\lambda,\widetilde{\eta},\widetilde{\delta})$ in the large $|\widetilde{\eta}|$ regime and
we show that, in the large  $|\widetilde{\eta}|$ limit,
$I^{(\R )}(\lambda,\widetilde{\eta},\widetilde{\delta})$ agrees with 
the limiting one-point function $I^{(\C )}(\lambda,\widetilde{\delta})$
of the complex Ginibre ensemble. We recall from~\cite[Eq. (13a)]{Cipolloni2020}
that the limit analogous to~\eqref{eq:largenas} for the complex case is given by
\begin{equation}
    \label{eq:IC}
    I^{(\C )}(\lambda,\widetilde{\delta}):=
    \frac{1}{2\pi\ii} \int\dif x\oint\dif y\;  e^{h(y)-h(x)} H(x,y)
\end{equation}
with
\begin{equation}
    \label{eq:newfunonreal}
    H(x,y):={}\frac{1}{x^3}+\frac{1}{x^2 y}+\frac{1}{xy^2}+\frac{\widetilde{\delta}}{xy}+\frac{\widetilde{\delta}}{x^2},\quad h(x):= -(1\wedge \widetilde\delta^{-1}) \lambda x+ \frac{\widetilde{\delta}}{x}+\frac{1}{2x^2}.
\end{equation}
The $x$-integration is over any contour from $0$ to $e^{3\ii \pi/4}\infty$, going out from $0$ in the direction of the positive real axis, and the $y$-integration is over any contour around $0$ in a counter-clockwise direction.
It is easy to see that along such contours the integral is absolutely convergent.
Note that the rhs.\ of~\eqref{eq:IC}  %
exactly agrees with~\cite[Eqs. (13a)--(13b)]{Cipolloni2020} after the change of variables $\widetilde{z}_*x\to x$ and $\widetilde{z}_*y\to y$, using the notation therein.

\begin{corollary}\label{theo:1realcompl}
    Let $I^{(\R )}(\lambda,\widetilde{\eta},\widetilde{\delta})$ be defined as in~\eqref{eq:exp1poi}, then it holds
    \begin{equation}
        \label{eq:largeetalim}
        \lim_{|\widetilde{\eta}|\to +\infty}I^{(\R )}(\lambda,\widetilde{\eta},\widetilde{\delta})
        =I^{(\C )}(\lambda,\widetilde{\delta}),
    \end{equation}
    for any fixed $\lambda\in \R _+$ and $\widetilde{\delta}=0$ or $\widetilde{\delta}\ge C_1$.
\end{corollary}

\begin{remark}
    From our analysis in  Section~\ref{sec:tildlim} (see~\eqref{eq:finetaco} later)
    it actually follows that $I^{(\R )}(\lambda,\widetilde{\eta},\widetilde{\delta})$ converges to $I^{(\C )}(\lambda,\widetilde{\delta})$ with a rate $|\widetilde{\eta}|^{-1}$. Similarly to Theorem~\ref{theo:1pointreal}, the convergence in~\eqref{eq:largeetalim} is uniform 
    in the entire range $\lambda\in [C_0^{-1}, C_0]$ and  $\widetilde{\delta}\in\{0\}\cup [ C_1,\infty)$ of the  other two parameters.
\end{remark}

\begin{remark}
    The limiting statement~\eqref{eq:largeetalim} follows by taking the $\widetilde{\eta}$
    limit within the formula~\eqref{eq:exp1poi}, i.e.\  %
    after the $N\to\infty$ limit is taken. However, we believe that in the regime $|\widetilde{\eta}|\ge C$, using a computation similar to the ones
    in Section~\ref{sec:tildlim} and to the bound~\cite[Lemmas 6.2-6.3-6.4]{Cipolloni2020}, but this time on the contours $\Lambda$, $\Omega$, one may prove the following stronger result:
    \begin{equation}
        \label{eq:onepointestreal}
        \begin{split}
            \E\frac{1}{N}\Tr(Y^z-\lambda\cdot c(N,\widetilde\delta)-\ii0)^{-1}&= N^{1/2}
            I^{(\C )}(\lambda,\widetilde{\delta}) 
            + \mathcal{O}\left(\left[(1\vee \widetilde{\delta})+\frac{1}{|\widetilde{\eta}|}\right] \bigl(1+\abs{\log \lambda}\bigr)\right).
        \end{split}
    \end{equation}
\end{remark}

\section{Derivation of the 1-point function}\label{sec:1p1pf}
Supersymmetric methods, especially the superbosonization formula (see e.g.~\cite{MR2430637}), provide an explicit formula for $\E \Tr [Y^z-w]^{-1}$. This was derived in~\cite[Eqs. (34)-(37)]{Cipolloni2020}, and with the choice $w=E+\ii \epsilon$ with $0<\epsilon\ll E\ll 1$, we have that
\begin{equation}
    \label{realsusyexplAAr}
    \E \Tr   [Y^z-w]^{-1}=\frac{N}{4\pi \ii} \oint \dif \xi \int_0^{\ii \infty} \dif a\int_0^1 \dif \tau \frac{\xi^2a}{\tau^{1/2}} e^{N[f(\xi,w)-g(a,\tau,\eta,w)]} G_N(a,\tau,\xi,\eta),
\end{equation}
where the $\xi$-integration is over any counter-clockwise oriented
contour around $0$ that does not encircle $-1$, the $a$,$\tau$-contours are straight lines, and, using the notation $\eta=\Im z$, the functions $f$ and $g$ are defined by
\begin{equation}
    \label{eq:fffr}
    f(\xi,w):= -w \xi+\log (1+\xi)-\log \xi-\frac{|z|^2}{1+\xi},      
\end{equation}
\begin{equation}
    \begin{split}
        \label{bosonphfAr}
        g(a,\tau,\eta,w)&:= -wa+\frac{1}{2}\log [1+2a+a^2\tau]-\log a-\frac{1}{2}\log \tau\\
        &\quad -\frac{|z|^2(1+a)-2\eta^2 a^2 (1-\tau)}{1+2a+a^2\tau}.
    \end{split}
\end{equation} 
The fact that the integral in~\eqref{realsusyexplAAr} is absolutely convergent follows by the explicit expressions of $f$ and $g$ in~\eqref{eq:fffr}-\eqref{bosonphfAr}. Note that $\epsilon$ in $w=E+\ii\epsilon$ is introduced only to make the 
$a$-integration on the imaginary axis absolutely convergent, hence, after the contours deformations described in Section~\ref{sec:cc} below, for all the practical purposes we can assume that $\epsilon=0$ and so $w=E$. Indeed, after deforming the $a$-contour so that it ends in the second quadrant, i.e.\ in the region $\{a\in\C :\Re[a]<0, \Im[a]>0\}$, we can take the limit $\epsilon\to 0^+$ since the integral in~\eqref{realsusyexplAAr} is absolutely convergent for $\epsilon=0$. Note that $g(a,1,\eta,w)=f(a,w)$; in particular, we remark that $g(a,1,\eta,w)$ is independent of $\eta$ for any $a\in\C  $. Furthermore, the function
\begin{equation}
    G_N(a,\tau,\xi,\eta):= G_{1,N}(a,\tau,\xi)+G_{2,N}(a,\tau,\xi,\eta)
\end{equation}
is given by
\begin{equation}
    \label{eq:newbetG}
    \begin{split}
        G_{1,N}={}&\Bigl(
        N^2\frac{p_{2,0,0}}{a^2 \xi ^2 (\xi +1)^2 \tau }-N\frac{ p_{1,0,0}}{a^2 \xi ^2 (\xi +1) \tau  }+\delta  N^2\frac{p_{2,0,1}}{a \xi  (\xi +1)^2 \tau  }- N \delta \frac{p_{1,0,1}}{a \xi  (\xi +1) \tau  } \\
        &\quad + N^2\delta ^2\frac{ p_{2,0,2}}{(\xi +1)^2}\Bigr)\times \Bigl((a^2\tau+2a+1)^2(\xi+1)^2\Bigr)^{-1}, \\
        G_{2,N}={}&\Bigl(
        N^2\eta ^2 \frac{p_{2,2,0}}{a \xi  (\xi +1)^3 \tau }-N\eta ^2 \frac{p_{1,2,0}}{a \xi  \tau  }+  N^2\eta ^2\delta \frac{p_{2,2,1}}{(\xi +1) }\Bigr)\times \Bigl((a^2\tau+2a+1)^2(\xi+1)^2\Bigr)^{-1},
    \end{split}
\end{equation}
where \(p_{i,j,k}=p_{i,j,k}(a, \tau, \xi)\) are explicit polynomials in \(a,\tau,\xi\) which we defer to Appendix~\ref{appendix poly} and \(\delta:=1 -\abs{z}^2\). The indices $i,j,k$ in the definition of $p_{i,j,k}$ denote the $N$, $\eta$ and $\delta$ power, respectively. We split $G_N$ as the sum of $G_{1,N}$ and $G_{2,N}$ since $G_{1,N}$ depends only on $|z|$, whilst $G_{2,N}$ depends explicitly on $\eta=\Im z$, in particular $G_{2,N}=0$ if $z\in \R  $.

\subsection{Choice of the integration contours}\label{sec:cc}
From now on we only focus on the regime $\widetilde{\delta}=0$, i.e.\ $|z|=1$. The proof in the case $\widetilde{\delta}\ge C_1$ for some large $C_1>0$ only requires slightly different choice of contours but otherwise the analysis of the integrals on them is analogous and so we omit the details.

\subsubsection{Geometry of \(\set{\Re g>0}\) in the regime \(\abs{\tau}\gtrsim 1\) (see Figure~\ref{fig:cont large})}
In the regime \(\abs{\tau}\gtrsim 1\) there is a transition at \(\abs{1-\tau}\eta^2 = E^{2/3}\). In the regime \(\abs{1-\tau}\eta^2\lesssim E^{2/3}\) there is only one relevant length scale of \(E^{-1/3}\). On the contrary in the regime \(\abs{1-\tau}\eta^2\gg E^{2/3}\) there are two relevant length scales \(E^{-1/3}\ll \abs{1-\tau}\eta^2 E^{-1}\), the former describing the size of the two connected components of \(\set{\Re g>0}\) close to \(0\) and the latter describing the distance to the infinite connected component of \(\set{\Re g>0}\) in the direction \(+\infty\). In Figure~\ref{fig:cont large} we present the level sets of \(\Re g\) for various sizes of \(\abs{1-\tau}\) and \(\eta\). 
\begin{figure}[htbp]
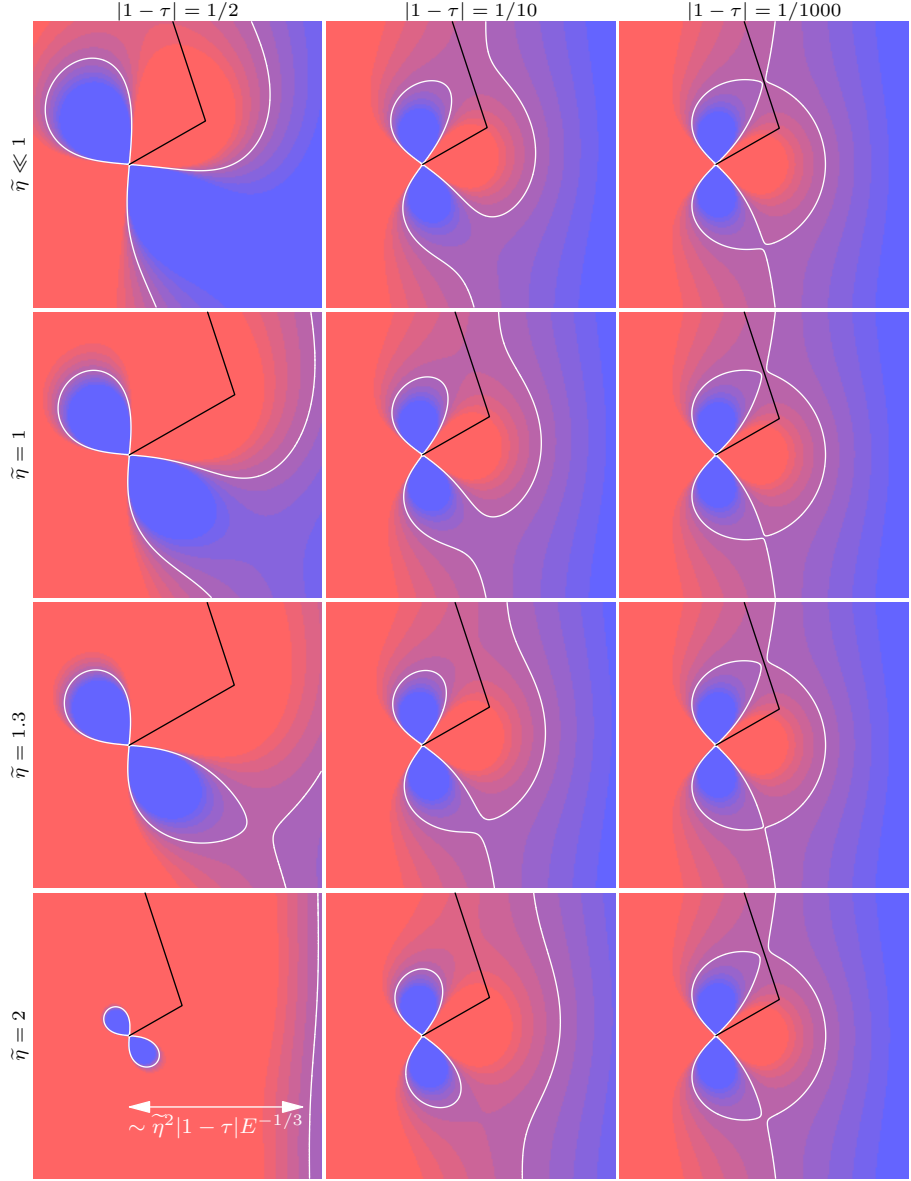

    \begin{asydef}
        pair xyMin=(-10000,-15000);
        pair xyMax=(20000,15000);
        var ee = 10.0^(-12.0);
        var ffS=-3*ee^(2/3)/4;
        var eta=0.;
    \end{asydef}
    \centering\footnotesize
    {\setlength{\tabcolsep}{0.12em}
    \renewcommand{\arraystretch}{.6}
    \begin{tabular}{cccc}
        & \(\abs{1-\tau}=1/2\)& \(\abs{1-\tau}=1/10\)& \(\abs{1-\tau}=1/1000\)\\
        \rotatebox{90}{\parbox{3.8cm}{\centering\(\wt\eta\ll 1\)}}& 
        \begin{asy}
            pair tau=tauC(.5);
            real r=sc(ee,tau,eta);
            real f(real x, real y) {return ff((x,y),ee,tau,eta) ;}
            picture bar; 
            bounds range=image(f,Range(ffS-ee^(2/3),ffS+ee^(2/3)),xyMin,xyMax,N,Palette);
            defaultpen(.5bp);
            draw(contour(f,xyMin,xyMax,new real[] {ffS},N,operator ..),white);
            draw((0,0)--aC(r,1)--aC(r,1 + (xyMax.y/(phiC*r) - sin(pi/6))/sin(3*pi/5)));
        \end{asy}
        
        &
        \begin{asy}
            pair tau=tauC(.9);
            real r=sc(ee,tau,eta);
            real f(real x, real y) {return ff((x,y),ee,tau,eta) ;}
            picture bar; 
            bounds range=image(f,Range(ffS-ee^(2/3),ffS+ee^(2/3)),xyMin,xyMax,N,Palette);
            defaultpen(.5bp);
            draw(contour(f,xyMin,xyMax,new real[] {ffS},N,operator ..),white);
            draw((0,0)--aC(r,1)--aC(r,1 + (xyMax.y/(phiC*r) - sin(pi/6))/sin(3*pi/5)));
        \end{asy}
        
        &
        \begin{asy}
            pair tau=tauC(.999);
            real r=sc(ee,tau,eta);
            real f(real x, real y) {return ff((x,y),ee,tau,eta) ;}
            picture bar; 
            bounds range=image(f,Range(ffS-ee^(2/3),ffS+ee^(2/3)),xyMin,xyMax,N,Palette);
            defaultpen(.5bp);
            draw(contour(f,xyMin,xyMax,new real[] {ffS},N,operator ..),white);
            draw((0,0)--aC(r,1)--aC(r,1 + (xyMax.y/(phiC*r) - sin(pi/6))/sin(3*pi/5))); 
        \end{asy} 
        
        \\ 
        \rotatebox{90}{\parbox{3.8cm}{\centering\(\wt\eta= 1\)}}&
        \begin{asy}
            real eta=10.^(-4.);
            pair tau=tauC(.5);
            real r=sc(ee,tau,eta);
            real f(real x, real y) {return ff((x,y),ee,tau,eta) ;}
            picture bar; 
            bounds range=image(f,Range(ffS-ee^(2/3),ffS+ee^(2/3)),xyMin,xyMax,N,Palette);
            defaultpen(.5bp);
            draw(contour(f,xyMin,xyMax,new real[] {ffS},N,operator ..),white);
            draw((0,0)--aC(r,1)--aC(r,1 + (xyMax.y/(phiC*r) - sin(pi/6))/sin(3*pi/5)));
        \end{asy}
        
        &
        \begin{asy}
            real eta=10.^(-4.);
            pair tau=tauC(.9);
            real r=sc(ee,tau,eta);
            real f(real x, real y) {return ff((x,y),ee,tau,eta) ;}
            picture bar; 
            bounds range=image(f,Range(ffS-ee^(2/3),ffS+ee^(2/3)),xyMin,xyMax,N,Palette);
            defaultpen(.5bp);
            draw(contour(f,xyMin,xyMax,new real[] {ffS},N,operator ..),white);
            draw((0,0)--aC(r,1)--aC(r,1 + (xyMax.y/(phiC*r) - sin(pi/6))/sin(3*pi/5)));
        \end{asy}
        
        &
        \begin{asy}
            real eta=10.^(-4.);
            pair tau=tauC(.999);
            real r=sc(ee,tau,eta);
            real f(real x, real y) {return ff((x,y),ee,tau,eta) ;}
            picture bar; 
            bounds range=image(f,Range(ffS-ee^(2/3),ffS+ee^(2/3)),xyMin,xyMax,N,Palette);
            defaultpen(.5bp);
            draw(contour(f,xyMin,xyMax,new real[] {ffS},N,operator ..),white);
            draw((0,0)--aC(r,1)--aC(r,1 + (xyMax.y/(phiC*r) - sin(pi/6))/sin(3*pi/5)));
        \end{asy} 
        
        \\
        \rotatebox{90}{\parbox{3.8cm}{\centering\(\wt\eta= 1.3\)}}& 
        \begin{asy}
            real eta=1.3*10.^(-4.);
            pair tau=tauC(.5);
            real r=sc(ee,tau,eta);
            real f(real x, real y) {return ff((x,y),ee,tau,eta) ;}
            picture bar; 
            bounds range=image(f,Range(ffS-ee^(2/3),ffS+ee^(2/3)),xyMin,xyMax,N,Palette);
            defaultpen(.5bp);
            draw(contour(f,xyMin,xyMax,new real[] {ffS},N,operator ..),white);
            draw((0,0)--aC(r,1)--aC(r,1 + (xyMax.y/(phiC*r) - sin(pi/6))/sin(3*pi/5)));
        \end{asy}
        
        &
        \begin{asy}
            real eta=1.3*10.^(-4.);
            pair tau=tauC(.9);
            real r=sc(ee,tau,eta);
            real f(real x, real y) {return ff((x,y),ee,tau,eta) ;}
            picture bar; 
            bounds range=image(f,Range(ffS-ee^(2/3),ffS+ee^(2/3)),xyMin,xyMax,N,Palette);
            defaultpen(.5bp);
            draw(contour(f,xyMin,xyMax,new real[] {ffS},N,operator ..),white);
            draw((0,0)--aC(r,1)--aC(r,1 + (xyMax.y/(phiC*r) - sin(pi/6))/sin(3*pi/5)));
        \end{asy}
        
        &
        \begin{asy}
            real eta=1.3*10.^(-4.);
            pair tau=tauC(.999);
            real r=sc(ee,tau,eta);
            real f(real x, real y) {return ff((x,y),ee,tau,eta) ;}
            picture bar; 
            bounds range=image(f,Range(ffS-ee^(2/3),ffS+ee^(2/3)),xyMin,xyMax,N,Palette);
            defaultpen(.5bp);
            draw(contour(f,xyMin,xyMax,new real[] {ffS},N,operator ..),white);
            draw((0,0)--aC(r,1)--aC(r,1 + (xyMax.y/(phiC*r) - sin(pi/6))/sin(3*pi/5)));
        \end{asy}
        
        \\
        \rotatebox{90}{\parbox{3.8cm}{\centering\(\wt\eta= 2\)}}& 
        \begin{asy}
            pair xyMin=(-20000.1,-30000);
            pair xyMax=(40000,30000);
            real eta=2*10.^(-4.);
            pair tau=tauC(.5);
            real r=sc(ee,tau,eta);
            real f(real x, real y) {return ff((x,y),ee,tau,eta) ;}
            picture bar; 
            bounds range=image(f,Range(ffS-ee^(2/3),ffS+ee^(2/3)),xyMin,xyMax,N,Palette);
            defaultpen(.5bp);
            draw(contour(f,xyMin,xyMax,new real[] {ffS},N,operator ..),white);
            draw("$\sim \widetilde\eta^2|1-\tau|E^{-1/3}$",(0,-15000)--(36000,-15000),arrow=Arrows,white);
            draw((0,0)--aC(r,1)--aC(r,1 + (xyMax.y/(phiC*r) - sin(pi/6))/sin(3*pi/5)));
        \end{asy} 
        
        &
        \begin{asy}
            pair xyMin=(-10000.1,-15000);
            pair xyMax=(20000,15000);
            real eta=2*10.^(-4.);
            pair tau=tauC(.9);
            real r=sc(ee,tau,eta);
            real f(real x, real y) {return ff((x,y),ee,tau,eta) ;}
            picture bar; 
            bounds range=image(f,Range(ffS-ee^(2/3),ffS+ee^(2/3)),xyMin,xyMax,N,Palette);
            defaultpen(.5bp);
            draw(contour(f,xyMin,xyMax,new real[] {ffS},N,operator ..),white);
            draw((0,0)--aC(r,1)--aC(r,1 + (xyMax.y/(phiC*r) - sin(pi/6))/sin(3*pi/5)));
        \end{asy}
        
        &
        \begin{asy}
            pair xyMin=(-10000.1,-15000);
            pair xyMax=(20000,15000);
            real eta=2*10.^(-4.);
            pair tau=tauC(.999);
            real r=sc(ee,tau,eta);
            real f(real x, real y) {return ff((x,y),ee,tau,eta) ;}
            picture bar; 
            bounds range=image(f,Range(ffS-ee^(2/3),ffS+ee^(2/3)),xyMin,xyMax,N,Palette);
            defaultpen(.5bp);
            draw(contour(f,xyMin,xyMax,new real[] {ffS},N,operator ..),white);
            draw((0,0)--aC(r,1)--aC(r,1 + (xyMax.y/(phiC*r) - sin(pi/6))/sin(3*pi/5)));
        \end{asy} 
        
    \end{tabular}}
    \caption{Contour plot of \(\Re g(\cdot,\tau,\wt\eta E^{1/3},E)\) for \(E>0\) for \(\tau\in\Omega\) with \(\abs{1-\tau}\le 1/2\). The white lines represent the level set \(\Re g(\cdot,\tau,\eta,E)=0\), while the black line represents the contour \(r\Lambda\) for the \(a\)-integration. All figures are on the same scale \(E^{-1/3}\), except for the bottom left figure which shows the larger scale \(E^{-1/3}\tilde\eta^2\abs{1-\tau}\), in addition to the \(E^{-1/3}\) length scale of the blue figure eight. The solid red colours are applied to regions where \(\Re g>E^{2/3}\), while the solid blue colours are applied to regions where \(\Re g<-E^{2/3}\).}\label{fig:cont large}
\end{figure}

\subsubsection{Geometry of \(\set{\Re g>0}\) in the regime \(\abs{\tau}\ll 1\) for \(\tilde\eta=0\) (see Figure~\ref{fig:cont small 0})}
In the regime \(\abs{\tau}\ll 1\) for \(\tilde\eta=0\) there is a transitions around \(\abs{\tau} = E\). For \(\abs{\tau}\ll E\) there are two components of \(\set{\Re g<0}\), one unbounded at a distance of \(E^{-1}\) to the right of the origin, and a bounded one at a distance of \(\abs{\tau}^{-1}\) below the origin. As \(\abs{\tau}\) approaches \(E\) the two components merge but remain separated from the origin at a distance of \(\abs{\tau}^{-2/3}E^{-1/3}\), see Figure~\ref{fig:cont small 0} for an illustration. 
\begin{figure}
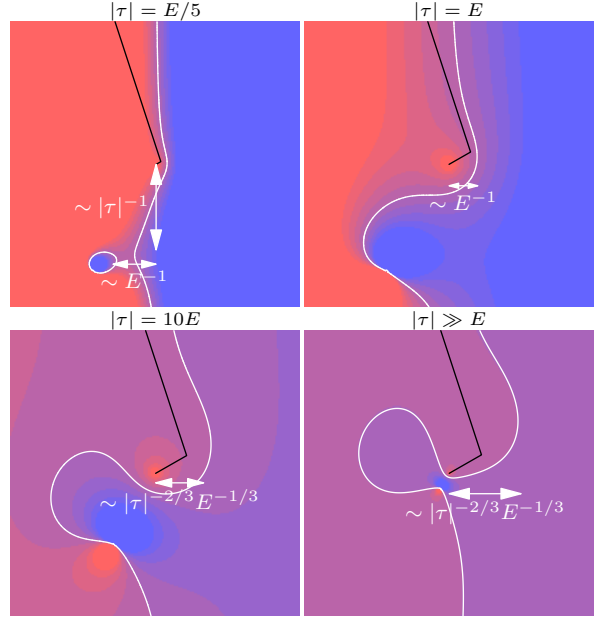

    \centering\footnotesize
    {\setlength{\tabcolsep}{0.12em}
    \renewcommand{\arraystretch}{.6}
    \begin{tabular}{cc}
        \(\abs{\tau}=E/5\)& \(\abs{\tau}=E\)\\
        \begin{asy}
            pair xyMin=(-1.01*10.^6,-10.^6);
            pair xyMax=(10.^6,10.^6);
            var ee = 10.0^(-5.0);
            var ffS=-3*ee^(2/3)/4;
            var eta=0.;
            real tt=ee/5;
            pair tau=tauC(tt);
            real r=sc(ee,tau,eta);
            real f(real x, real y) {return ff((x,y),ee,tau,eta) ;}
            picture bar; 
            bounds range=image(f,Range(-1,1),xyMin,xyMax,N,Palette);
            defaultpen(.5bp);
            draw(contour(f,xyMin,xyMax,new real[] {ffS},N,operator ..),white);
            draw((0,0)--aC(r,1)--aC(r,1 + (xyMax.y/(phiC*r) - sin(pi/6))/sin(3*pi/5)));
            draw("$\sim |\tau|^{-1}$",(0,0)--(0,-.6*10.^6),arrow=Arrows,white);
            draw("$\sim E^{-1}$",(-.3*10.^6,-.7*10.^6)--(0,-.7*10.^6),arrow=Arrows,white);
        \end{asy}
        
        &
        
        \begin{asy}
            pair xyMin=(-2.01*10.^5,-2*10.^5);
            pair xyMax=(2*10.^5,2*10.^5);
            var ee = 10.0^(-5.0);
            var ffS=-3*ee^(2/3)/4;
            var eta=0.;
            real tt=ee;
            pair tau=tauC(tt);
            real r=sc(ee,tau,eta);
            real f(real x, real y) {return ff((x,y),ee,tau,eta) ;}
            picture bar; 
            bounds range=image(f,Range(-1,1),xyMin,xyMax,N,Palette);
            defaultpen(.5bp);
            draw(contour(f,xyMin,xyMax,new real[] {ffS},N,operator ..),white);
            draw((0,0)--aC(r,1)--aC(r,1 + (xyMax.y/(phiC*r) - sin(pi/6))/sin(3*pi/5)));
            draw("$\sim E^{-1}$",(0,-.3*10.^5)--(.4*10.^5,-.3*10.^5),arrow=Arrows,white);
        \end{asy}
        
        \\
        \(\abs{\tau}=10E\)& \(\abs{\tau}\gg E\)\\
        \begin{asy}
            pair xyMin=(-3.01*10.^4,-3*10.^4);
            pair xyMax=(3*10.^4,3*10.^4);
            var ee = 10.0^(-5.0);
            var ffS=-3*ee^(2/3)/4;
            var eta=0.;
            real tt=10*ee;
            pair tau=tauC(tt);
            real r=sc(ee,tau,eta);
            real f(real x, real y) {return ff((x,y),ee,tau,eta) ;}
            picture bar; 
            bounds range=image(f,Range(-1,1),xyMin,xyMax,N,Palette);
            defaultpen(.5bp);
            draw(contour(f,xyMin,xyMax,new real[] {ffS},N,operator ..),white);
            draw((0,0)--aC(r,1)--aC(r,1 + (xyMax.y/(phiC*r) - sin(pi/6))/sin(3*pi/5)));
            pair z2=1.5*aC(r,1) + (0,-.4*10.^4);
            draw("$\sim |\tau|^{-2/3}E^{-1/3}$",(0,-.2*10.^4)--(1*10.^4,-.2*10.^4),arrow=Arrows,white);
        \end{asy} 
        
        &
        \begin{asy}
            pair xyMin=(-1.41*10.^3,-1.4*10.^3);
            pair xyMax=(1.4*10.^3,1.4*10.^3);
            var ee = 10.0^(-5.0);
            var ffS=-3*ee^(2/3)/4;
            var eta=0.;
            real tt=1000*ee;
            pair tau=tauC(tt);
            real r=sc(ee,tau,eta);
            real f(real x, real y) {return ff((x,y),ee,tau,eta) ;}
            picture bar; 
            bounds range=image(f,Range(-1,1),xyMin,xyMax,N,Palette);
            defaultpen(.5bp);
            draw(contour(f,xyMin,xyMax,new real[] {ffS},N,operator ..),white);
            draw((0,0)--aC(r,1)--aC(r,1 + (xyMax.y/(phiC*r) - sin(pi/6))/sin(3*pi/5)));
            draw("$\sim |\tau|^{-2/3}E^{-1/3}$",(0,-.2*10.^3)--(.7*10.^3,-.2*10.^3),arrow=Arrows,white);
        \end{asy} 
        
    \end{tabular}}
    \caption{Contour plot of \(\Re g(\cdot,\tau,0,E)\) for \(E>0\) for \(\tau\in\Omega\) with \(\abs{\tau}\ll 1\). The solid white lines represent the level set \(\Re g(\cdot,\tau,0,E)=0\), while the solid black line represents the contour \(r\Lambda\) for the \(a\)-integration. The solid red colours are applied to regions where \(\Re g>1\), while the solid blue colours are applied to regions where \(\Re g<-1\)}\label{fig:cont small 0}
\end{figure}

\subsubsection{Geometry of \(\set{\Re g>0}\) in the regime \(\abs{\tau}\ll 1\) for \(\tilde\eta>0\) (see Figure~\ref{fig:cont small})}
For \(\abs{\tau}\ll 1\) and \(\tilde\eta>0\) there are two components of \(\set{\Re g<0}\), one unbounded one at a distance of \(\abs{\tau}^{-1}E^{-1/3}\) to the right of the origin, and a bounded one at the bottom left of the origin. The bounded component is an approximate disk of diameter \(\abs{\tau}^{-1}\) for \(\abs{\tau}\ll E^{2/3}\) and is transformed into a ``lying eight'' of diameter \(\abs{\tau}^{-1/2}E^{-1/3}\) as \(\abs{\tau}\gg E^{2/3}\), see Figure~\ref{fig:cont small} for an illustration. 
\begin{figure}
    \centering\footnotesize
    {\setlength{\tabcolsep}{0.12em}
    \renewcommand{\arraystretch}{.6}
    \begin{tabular}{cc}
        \(\abs{\tau}\ll E^{2/3}\)&\(\abs{\tau}=E^{2/3}\)\\        
        \begin{asy}
            pair xyMin=(-.51,-1.01);
            pair xyMax=(1.5,1.);
            var ee = 10.^(-6.);
            var ffS=-3*ee^(2/3)/4;
            var eta=ee^(1./3);
            real tt=10.^(-7.);
            pair tau=tauC(tt);
            real r=sc(ee,tau,eta);
            real f(real x, real y) {return max(-1,min(1,ff((r*x,r*y),ee,tau,eta))) ;}
            picture bar;
            bounds range=image(f,Range(-1,1),xyMin,xyMax,N,Palette);
            defaultpen(.5bp);
            draw(contour(f,xyMin,xyMax,new real[] {ffS},N,operator ..),white);
            draw((0,0)--aC(r,1)/r--aC(r,1 + (xyMax.y/(phiC) - sin(pi/6))/sin(3*pi/5))/r);
            draw("$|\tau|^{-1}$",(-.08,-.1)--(.02,-.1),arrow=Arrows,white);
            draw("$|\tau|^{-1}E^{-1/3}$",(0,-.6)--(1.,-.6),arrow=Arrows,white);
        \end{asy} 
        
        & 
        \begin{asy}
            pair xyMin=(-.51,-1.01);
            pair xyMax=(1.5,1.);
            var ee = 10.^(-6.);
            var ffS=-3*ee^(2/3)/4;
            var eta=ee^(1./3);
            real tt=ee^(2/3);
            pair tau=tauC(tt);
            real r=sc(ee,tau,eta);
            real f(real x, real y) {return max(-1,min(1,ff((r*x,r*y),ee,tau,eta))) ;}
            picture bar;
            bounds range=image(f,Range(-1,1),xyMin,xyMax,N,Palette);
            defaultpen(.5bp);
            draw(contour(f,xyMin,xyMax,new real[] {ffS},N,operator ..),white);
            draw((0,0)--aC(r,1)/r--aC(r,1 + (xyMax.y/(phiC) - sin(pi/6))/sin(3*pi/5))/r);
            draw("$|\tau|^{-1}$",(-.08,-.1)--(.02,-.1),arrow=Arrows,white);
            draw("$|\tau|^{-1}E^{-1/3}$",(0,-.6)--(1.,-.6),arrow=Arrows,white);
        \end{asy}
        
        \\\(\abs{\tau}=100E^{2/3}\)&\(1\gg\abs{\tau}\gg E^{2/3}\)\\
        \begin{asy}
            pair xyMin=(-.51,-1.01);
            pair xyMax=(1.5,1.);
            var ee = 10.^(-6.);
            var ffS=-3*ee^(2/3)/4;
            var eta=ee^(1./3);
            real tt=100.*ee^(2/3);
            pair tau=tauC(tt);
            real r=sc(ee,tau,eta);
            real f(real x, real y) {return max(-1,min(1,ff((r*x,r*y),ee,tau,eta))) ;}
            picture bar;
            bounds range=image(f,Range(-1,1),xyMin,xyMax,N,Palette);
            defaultpen(.5bp);
            draw(contour(f,xyMin,xyMax,new real[] {ffS},N,operator ..),white);
            draw((0,0)--aC(r,1)/r--aC(r,1 + (xyMax.y/(phiC) - sin(pi/6))/sin(3*pi/5))/r);
            draw("$|\tau|^{-1/2}E^{-1/3}$",(-.1,-.15)--(.1,-.15),arrow=Arrows,white);
            draw("$|\tau|^{-1}E^{-1/3}$",(0,-.7)--(1.,-.7),arrow=Arrows,white);
        \end{asy}
        
        &
        \begin{asy}
            pair xyMin=(-.51,-1.01);
            pair xyMax=(1.5,1.);
            var ee = 10.^(-6.);
            var ffS=-3*ee^(2/3)/4;
            var eta=ee^(1./3);
            real tt=600.*ee^(2/3);
            pair tau=tauC(tt);
            real r=sc(ee,tau,eta);
            real f(real x, real y) {return max(-1,min(1,ff((r*x,r*y),ee,tau,eta))) ;}
            picture bar;
            bounds range=image(f,Range(-1,1),xyMin,xyMax,N,Palette);
            defaultpen(.5bp);
            draw(contour(f,xyMin,xyMax,new real[] {ffS},N,operator ..),white);
            draw((0,0)--aC(r,1)/r--aC(r,1 + (xyMax.y/(phiC) - sin(pi/6))/sin(3*pi/5))/r);
            draw("$|\tau|^{-1/2}E^{-1/3}$",(-.25,-.3)--(.3,-.3),arrow=Arrows,white);
            draw("$|\tau|^{-1}E^{-1/3}$",(0,-.7)--(.85,-.7),arrow=Arrows,white);
        \end{asy}

    \end{tabular}}
    \caption{Contour plot of \(\Re g(\cdot,\tau,\tilde \eta E^{1/3},E)\) for \(0\le E\ll1\) and \(\tilde \eta>0\) for \(\tau\in\Omega\) with \(\abs{\tau}\ll 1\). The solid white lines represent the level set \(\Re g(\cdot,\tau,\tilde\eta E^{1/3},E)=0\), while the solid black line represents the contour \(r\Lambda\) for the \(a\)-integration. The solid red colours are applied to regions where \(\Re g>1\), while the solid blue colours are applied to regions where \(\Re g<-1\).}\label{fig:cont small}
\end{figure}

\subsubsection{Deformation of contours} 
Now we explain how the contours in~\eqref{realsusyexplAAr} can be deformed. The $\xi$-contour can be freely deformed as long as it does not cross $0$ and $-1$. We can deform the $\tau$-contour as long as $\Im[\tau]<0$, then the $a$-contour has to be deformed accordingly to ensure the absolute convergence of the integral. The $a$-contour at infinity can be freely deformed, independently of $\tau$, as long as it ends in the second quadrant; on the other hand the way how it goes out from zero depends on $\tau$. Moreover along the deformation of the $a$-contour we cannot cross the points $(-1\pm \sqrt{1-\tau})\tau^{-1}$ which are the singularities of the term $a^2\tau+2a+1$ in $g$ and $G_N$. In particular, note that the $\tau$ and $a$ contours cannot be deformed independently: we first deform the $\tau$-contour and then we deform the $a$-contour accordingly. In the remainder of this section we will always deform the integration contours as described above.

Next, we describe how we concretely deform the integration contours in~\eqref{realsusyexplAAr}
in accordance with the rules just described. From now on we denote the $\xi$-contour by $\Gamma$, the $\tau$-contour by $\Omega$, and the $a$-contour by $\Lambda$. In particular, we choose
\begin{equation}
    \label{eq contour choice}
    \begin{split}
        \Gamma&:=-\frac{1}{2}+E^{-1/3}+E^{-1/3}\partial\DD,\\
        \Omega&:= \Bigl[0,\frac{1}{3\sqrt{3}}-\frac{\ii}{3}\Bigr] \cup \Bigl(\frac{1}{3\sqrt{3}}-\frac{\ii}{3},1-\frac{1}{3\sqrt{3}}-\frac{\ii}{3}\Bigr]\cup \Bigl(1-\frac{1}{3\sqrt{3}}-\frac{\ii}{3},1\Bigr],\\
        \Lambda&:=\big[0,z_0\big]\cup\Bigl(z_0, z_0+ e^{3\ii\pi/5}\infty\Bigr), \qquad  z_0: = \frac{\sin(4\pi/15)}{\sin(17\pi/30)} e^{\ii \pi/6},
    \end{split}
\end{equation}
and  rescale \(r\Lambda\) with a parameter \(r>0\) chosen later depending on \(\widetilde{\eta},\tau,E\). Here the interval \((0,e^{3\ii\pi/5}\infty)\) is understood as the open half-line going out to \(\infty\) in the \(e^{3\ii\pi/5}\) direction. Note that the contour \(\Lambda\) is designed such that \(e^{\ii\pi/3}\in\Lambda\). According to the geometry of \(\set{\Re g<0}\) we choose the scaling parameter 
\begin{equation}\label{r def}
    r=r(\wt\eta,E,\tau):= \frac{1}{E^{1/3}}\biggl( \frac{1}{2}+\frac{1}{2(E\vee\abs{\tau})^{2/3}}+\frac{\tilde\eta^2}{\abs{\tau}} \biggr).
\end{equation}

We can thus rewrite~\eqref{realsusyexplAAr} as
\begin{equation}
    \label{eq:finnewform}
    \E \Tr   [Y^z-w]^{-1}=\frac{N}{4\pi \ii} \oint_\Gamma \dif \xi \int_\Omega \dif \tau \int_{r\Lambda} \dif a \frac{\xi^2a}{\tau^{1/2}} e^{N[f(\xi,w)-g(a,\tau,\eta,w)]} G_N(a,\tau,\xi,\eta).
\end{equation}
In the following we split the computation of the leading term of~\eqref{eq:finnewform} into two parts: (i) in Section~\ref{sec:smalll} we deal with the regime when either $|\xi|\le N^\omega$ or $|a\tau|\le N^{2\omega}$ or $|\tau|\le N^{-\omega}$, for some small fixed $\omega>0$, (ii) in Section~\ref{sec:large} we deal with the complementary regime when $|\xi|$ and $|a\tau|$ are bigger than $N^\omega$ and $|\tau|>N^{-\omega}$.

\subsection{Small \texorpdfstring{$|\xi|$}{|xi|} or small \texorpdfstring{$|a\tau|$}{|a tau|} or small \texorpdfstring{$|\tau|$}{|tau|} regime}\label{sec:smalll}
By the explicit form of the phase functions $f(\xi,E)$ and $g(a,\tau,\eta,E)$ in~\eqref{eq:fffr}-\eqref{bosonphfAr} it follows that the contribution to~\eqref{realsusyexplAAr} of the small $|\xi|$ and $|a\tau|$ regimes is negligible; in particular, the smallness comes from the logarithmic factors in the phase functions. This is made rigorous in Lemma~\ref{lem:smallatau}. For this purpose we define the contours 
\begin{equation}
    \label{eq:smallreg}
    \widetilde{\Gamma}:= \{\xi\in \Gamma : |\xi|\le N^\omega\}, \qquad \widetilde{\Lambda}:= \{a\in r\Lambda : |a|\le N^{2\omega} |\tau|^{-1}\}, \qquad \widetilde{\Omega}:=\{\tau\in \Omega : |\tau|\le N^{-\omega}\},
\end{equation}
for some small fixed $\omega>0$. Note that $r\Lambda\setminus \widetilde{\Lambda}$ is always connected.

\begin{lemma}\label{lem:smallatau}
    Let $f,g, G_N$ be defined in~\eqref{eq:fffr}-\eqref{eq:newbetG}, and let $\widetilde{\Gamma}$, $\widetilde{\Omega}$, $\widetilde{\Lambda}$ be the contours defined in~\eqref{eq:smallreg}, then for any large constant $C_4>0$, for any $E=\lambda N^{-3/2}$, with $C_4^{-1}\le\lambda\le C_4$, and for any $\widetilde{\eta}=0$ or $C_4^{-1}\le|\widetilde{\eta}|\le C_4$, we have that 
    \begin{equation}
        \label{eq:irrreg}
        \begin{split}
            &\left|\left(\oint_\Gamma \dif \xi\int_\Omega\dif \tau\int_{r\Lambda} \dif a-\oint_{\Gamma\setminus\widetilde{\Gamma}} \dif \xi\int_{\Omega\setminus\widetilde{\Omega}}\dif \tau\int_{r\Lambda\setminus \widetilde{\Lambda}} \dif a\right)\left[e^{N[f(\xi,E)-g(a,\tau,\eta,E)]} \frac{a\xi^2}{\tau^{1/2}} G_N(a,\tau,\xi,z)\right] \right| \\
            &\qquad \le C e^{-N^{\omega/10}}. 
        \end{split}
    \end{equation}
    The constant $C>0$ only depends on $C_4$.
\end{lemma}
\begin{proof}
    The proof relies on two quantitative lower bounds on \(\Re g\) outlined in the following lemmata, the proofs of which we defer to Appendix~\ref{app:addlemm}. Within these Lemmas and their proofs 
    we deviate from our general convention and
    the notation $f\ll g$ means that $f\le cg$ for a sufficiently small
    $N$-independent constant $c$. 
    \begin{lemma}\label{lemma lower bound g}
        For \(\abs{\tau}\le N^{-\epsilon}\) we have the following lower bound on \(\Re g\) which for clarity we formulate separately depending on the relative sizes of \(\tau,E,a\) and whether \(\wt\eta=0\) or \(\ne 0\).  
        \setlist[enumerate,1]{label={\bfseries\arabic*.}}
        \setlist[enumerate,2]{label={\bfseries\alph*)},ref={\bfseries\theenumi\alph*)}}
        \begin{enumerate}[itemsep=6pt]
            \item \(\abs{\tau}\lesssim E\) and \(\tilde\eta=0\), hence \(r\sim E^{-1}\).
            \begin{enumerate}
                \item\label{case1} \makebox[5cm][l]{\(\abs{a}\le \abs{\tau}^{-1}\):}\(\Re g\gtrsim 1-\log\abs{a\tau}\)
                \item\label{case2} \makebox[5cm][l]{\(\abs{\tau}^{-1}<\abs{a}\lesssim E^{-1} \):}\(\Re g\gtrsim 1\)
                \item\label{case2b} \makebox[5cm][l]{\(\abs{a}\gg E^{-1}\):}\(\Re g\gtrsim E\abs{a}.\)
            \end{enumerate}
            \item \(\abs{\tau}\lesssim E\) and \(\tilde\eta\ne 0\), hence \(r\sim E^{-1/3}\abs{\tau}^{-1}\).
            \begin{enumerate}
                \item\label{case3} \makebox[5cm][l]{\(\abs{a}\lesssim E^{-1}\wedge \abs{\tau}^{-1}\):}\(\Re g\gtrsim 1-\log\abs{a\tau}\)
                \item\label{case4} \makebox[5cm][l]{\(E^{-1}\wedge \abs{\tau}^{-1}\ll \abs{a}\le \abs{\tau}^{-1}\):}\(\Re g\gtrsim E^{2/3}\wt\eta^2\abs{a}-\log\abs{a\tau}\)
                \item\label{case5} \makebox[5cm][l]{\(\abs{\tau}^{-1}<\abs{a}\lesssim\abs{\tau}^{-1}E^{-1/3}\):}\(\Re g\gtrsim E^{2/3}\wt\eta^2\abs{\tau}^{-1}\)
                \item\label{case6} \makebox[5cm][l]{\(\abs{a}\gg E^{-1/3}\abs{\tau}^{-1}\):}\(\Re g\gtrsim E\abs a.\)
            \end{enumerate}
            \item \(E\ll \abs{\tau}\le N^{-\epsilon}\) and \(\wt \eta=0\), hence \(r\sim E^{-1/3}\abs{\tau}^{-2/3}\).
            \begin{enumerate}
                \item\label{case7} \makebox[5cm][l]{\(\abs{a}\ll \abs{\tau}^{-1}\):}\(\Re g\gtrsim -\log\abs{a\tau}\)
                \item\label{case8} \makebox[5cm][l]{\(\abs{\tau}^{-1}\lesssim\abs{a}\lesssim E^{-1/3}\abs{\tau}^{-2/3}\):}\(\Re g\gtrsim E^{2/3}\abs{\tau}^{-2/3}\)
                \item\label{case9} \makebox[5cm][l]{\(\abs{a}\gg E^{-1/3}\abs{\tau}^{-2/3}\):}\(\Re g\gtrsim E \abs a.\)
            \end{enumerate}
            \item \(E\ll \abs{\tau}\le N^{-\epsilon}\) and \(\wt \eta\ne 0\), hence \(r\sim E^{-1/3}\abs{\tau}^{-1}\).
            \begin{enumerate}
                \item\label{case10} \makebox[5cm][l]{\(\abs{a}\ll \abs{\tau}^{-1}\):}\(\Re g\gtrsim -\log\abs{a\tau}\)
                \item\label{case11} \makebox[5cm][l]{\(\abs{\tau}^{-1}\lesssim\abs{a}\lesssim E^{-1/3}\abs{\tau}^{-1}\):}\(\Re g\gtrsim E^{2/3}\wt\eta^2\abs{\tau}^{-1}\)
                \item\label{case12} \makebox[5cm][l]{\(\abs{a}\gg E^{-1/3}\abs{\tau}^{-1}\):}\(\Re g\gtrsim E \abs a.\)
            \end{enumerate}
        \end{enumerate}
    \end{lemma}
    \begin{lemma}\label{lem:dec}
        For any \(1\ge\abs{\tau}\gg E\) with \(\tau\in\Omega\) the function 
        \begin{equation}
            \label{eq:dec}
            x \mapsto \Re g(x e^{\ii\pi/6},\tau,0,E)
        \end{equation}
        is monotonically decreasing in \(x\) for \(0\le x\ll E^{-1/3}\). Moreover, for any \(\eta\ge 0\), and any \(1\ge\abs{\tau}\gg E\), \(0\le x\ll E^{-1/3}\) we have 
        \begin{equation}
            \label{eq:etanoimp}
            \Re g(x e^{\ii\pi/6},\tau,\eta,E) \ge \Re g(x e^{\ii\pi/6},\tau,0,E).
        \end{equation}
    \end{lemma}
    
    We now split the proof of~\eqref{eq:irrreg} into three parts, we first prove that the contribution to~\eqref{eq:finnewform} in the regime $\tau\in \widetilde{\Omega}$ is exponentially small uniformly in $\xi\in \Gamma$ and $a\in r\Lambda$. Then we prove that the regime $a\in\widetilde{\Lambda}$ is also exponentially small for any $\xi\in \Gamma$ and $\tau \in \Omega\setminus \widetilde{\Omega}$. Finally, we conclude that also the contribution for $\xi\in \widetilde{\Gamma}$ is negligible.
    
    We start with the regime $\tau\in \widetilde{\Omega}$. Similarly to~\cite[Eq. (97)]{Cipolloni2020}, using that $|1+2a+a^2\tau|\gtrsim 1$, we have that
    \begin{equation}
        \label{eq:impb}
        \left|\oint_\Gamma G_N(a,\tau,\xi,z)\xi^2 e^{Nf(\xi)}\dif \xi\right|\lesssim N^3 \left(1+\frac{1}{|a|}+\frac{1}{|a|^2|\tau|}\right).
    \end{equation}
    
    Then, given, the lower bounds for $\Re g$ in~\ref{case1}--\ref{case12} by simple computations we conclude the following lemma.
    
    \begin{lemma}
        For any $\alpha,\gamma\in \R $ it holds
        \begin{equation}
            \label{eq:expsmall}
            \int_{\widetilde{\Omega}}|\dif\tau|\int_{r\Lambda}|\dif a|\, |a|^\alpha |\tau|^{-\gamma} e^{-\Re g(a,\tau,\eta,E)}\le N^{C(\alpha,\gamma)} e^{-N^{\omega/10}},
        \end{equation}
        for some $N$-independent constant $C(\alpha,\gamma)>0$.
    \end{lemma}
    
    Using the bound in~\eqref{eq:expsmall} we readily conclude that the contribution of the regime $\tau\in \widetilde{\Omega}$ is exponentially small and so negligible.
    
    We now consider the regime $a\in\widetilde{\Lambda}$. We split this regime into two cases: (i) $|a|\ge N^{-10}$, (ii) $|a|\le N^{-10}$. For $|a|\ge N^{-10}$, by~\eqref{eq:impb} and Lemma~\ref{lem:dec}, we readily conclude that
    \begin{equation}
        \label{eq:smallab}
        \int_{\Omega\setminus\widetilde{\Omega}}|\dif\tau|\int_{\widetilde{\Lambda}}|\dif a|\, |a|^\alpha |\tau|^{-\gamma} e^{-\Re g(a,\tau,\eta,E)}\le N^{C(\alpha,\gamma)} e^{-N^{1-2\omega}}.
    \end{equation}
    In the regime $|a|\le N^{-10}$ we conclude a bound as in~\eqref{eq:smallab} using the explicit form of $g$ in~\eqref{bosonphfAr} and that $|\tau|\le 1$. This proves that also the regime $a\in\widetilde{\Lambda}$ is negligible.
    
    Finally, the fact that the regime $\xi\in\widetilde{\Gamma}$ is exponentially small, given that both the regimes $\tau\in\widetilde{\Omega}$ and $a\in \widetilde{\Lambda}$ are removed, follows exactly as in the proof~\cite[Lemma 6.4]{Cipolloni2020}.
\end{proof}

\subsection{The regime where  \texorpdfstring{$|\xi|$}{|xi|},  \texorpdfstring{$|a\tau|$}{|a tau|} and
\texorpdfstring{$|\tau|$}{|tau|} are all large}\label{sec:large}
In the remainder of this section we focus on the regime when $|\xi|\ge N^\omega$, $|a|\ge N^{2\omega}|\tau|^{-1}$ and $|\tau|\le N^{-\omega}$, and in this regime we expand $f(\cdot,E), g(\cdot,\tau,\eta,E), G_N$ similarly to Eq. (75)-(77) of~\cite{Cipolloni2020}. By Taylor expansion for large $|\xi|$ and large $|a\tau|$ we have
\begin{equation}
    \label{eq:largregexpxi}
    \begin{split}
        f(\xi,E)&=\left[-E \xi+\frac{1}{2\xi^2}\right]\times \left(1+\mathcal{O}\big(|\xi|^{-1}\big)\right) \\
        g(a,\tau,\eta,E)&=\left[-Ea-\frac{2 \eta ^2 (\tau -1)}{\tau }+\frac{\delta}{a\tau}+ \frac{(2-\tau)}{2 a^2 \tau ^2}+\frac{2\eta^2(\tau-5)}{a^2\tau^2}\right]\times\left(1+\mathcal{O}\left(|a\tau|^{-1}\right)\right),
    \end{split}
\end{equation}
and\footnote{Note that the term $-\frac{N^2}{a^4\tau^2\xi^4}$ was erroneously missing in the expansion in~\cite[Eq. (76)]{Cipolloni2020}.}
\begin{align}
    G_{1,N}(a,\tau,\xi,|z|)&= \Bigg[\sum_{\substack{a,\beta\ge 2, \, \alpha+\beta=8,\\\gamma=\min\{\alpha-1,3\}}} \frac{c_{1,\alpha,\beta,\gamma}N^2}{a^\alpha \tau^\gamma  \xi^{\beta}}-\frac{N}{a^4\tau^2\xi^4}+\sum_{\substack{\alpha,\beta\ge 2, \, \alpha+\beta=7, \\ \gamma=\min\{\alpha-1,3\}}}\frac{c_{2,\alpha,\beta,\gamma}N^2\delta}{a^\alpha \tau^\gamma\xi^\beta} \nonumber\\\label{expfxi2AAr}
    &\quad+\sum_{\substack{a,\beta\ge 2,\, \alpha+\beta=6,\\ \gamma=\min\{\alpha-1,2\}}} \frac{c_{3,\alpha,\beta,\gamma}N}{a^\alpha \tau^\gamma  \xi^{\beta}} +\sum_{\substack{\alpha, \beta\ge 2, \alpha+\beta=6 \\ \gamma=\min\{\alpha-1,2\}}}\frac{c_{4,\alpha,\beta,\gamma} N^2\delta^2}{a^\alpha\tau^\gamma\xi^\beta}\Bigg]\times\big[1+\mathcal{O}\big(|a\tau|^{-1}+|\xi|^{-1}\big)\big], \\
    G_{2,N}(a,\tau,\xi,z)&=\Bigg[\sum_{\substack{\alpha,\beta\ge 2, \, \alpha+\beta=6, \\ \gamma=\max\{ \alpha-1,2\}}} \frac{4 N^2\eta^2}{a^\alpha\tau^\gamma\xi^{\beta}} +\sum_{\substack{\alpha,\beta\ge 2, \alpha+\beta=5, \\ \gamma=\max\{\alpha-1,2\}}}\frac{4 N^2\eta^2\delta}{a^\alpha\tau^\gamma\xi^\beta} \nonumber\\ \label{newexpansionA1Ar}
    &\quad + \sum_{\substack{\alpha,\beta=2,3,\, \alpha+\beta=5 \\ \gamma=\max\{\alpha-1,2\}}} \frac{4 N\eta^2}{a^\alpha\tau^\gamma\xi^{\beta}} \Bigg]\times\big[1+\mathcal{O}\big(|a\tau|^{-1}+|\xi|^{-1}\big)\big],
\end{align}
where $c_{i,\alpha,\beta,\gamma}\in \R $ are defined as in Appendix~\ref{app:expco}.

\subsection{Proof of Theorem~\ref{theo:1pointreal}}
We recall that we only prove the case $\widetilde{\delta}=0$; the case $\widetilde{\delta}\ge C_1$ is completely analogous and so omitted. By Lemma~\ref{lem:smallatau} and~\eqref{eq:finnewform} we conclude that
\begin{equation}
    \label{eq:nosmall}
    \E \Tr   [Y^z-w]^{-1}=\frac{N}{4\pi \ii} \oint_{\Gamma\setminus \widetilde{\Gamma}} \dif \xi \int_{\Omega\setminus \widetilde{\Omega}} \dif \tau \int_{r\Lambda\setminus \widetilde{\Lambda}} \dif a \frac{\xi^2a}{\tau^{1/2}} e^{N[f(\xi,w)-g(a,\tau,\eta,w)]} G_N(a,\tau,\xi,\eta)
\end{equation}
up to an exponentially small error that we will always ignore in the sequel. In order to compute the leading order of~\eqref{eq:nosmall} as $N$ goes to infinity, we use the change of variables
\begin{equation}
    \label{eq:cv}
    \widetilde{\eta}=N^{1/2}\eta, \quad \lambda=N^{3/2}E, \quad a'=a N^{-1/2}, \quad \xi'=\xi N^{-1/2},
\end{equation}
where $a', \xi'$ are the new integration variables. We get that (omitting the primes, i.e.\ 
using the notation $a$, $\xi$ for the new variables as well to make the notation simpler)
\begin{equation}
    \label{eq:nosmallafterchange}
    \begin{split}
        &\E \Tr   [Y^z-w]^{-1} \\
        &\quad=\frac{N^{3/2}}{4\pi \ii} \int_{N^{-1/2}(\Gamma\setminus \widetilde{\Gamma})} \dif \xi \int_{\Omega\setminus \widetilde{\Omega}} \dif \tau \int_{N^{-1/2}(r\Lambda\setminus \widetilde{\Lambda})} \dif a \frac{\xi^2a}{\tau^{1/2}} e^{N[\f(\xi,w)-\g(a,\tau,\eta,w)]} G(a,\tau,\xi,\widetilde{\eta})\\
        &\qquad+\mathcal{O}(N).
    \end{split}
\end{equation}
Here we used the asymptotic relations
\begin{equation}
    \label{eq:leadterm}
    \begin{split}
        f(\sqrt{N}\xi,\lambda)&=N^{-1}\f(\xi,\lambda)\left( 1+\mathcal{O}\left(\frac{1}{N^{1/2}|\xi|}\right)\right), \\
        g(\sqrt{N} a,\tau,\sqrt{N}\widetilde{\eta},\lambda)&=N^{-1}\g(a,\tau,\widetilde{\eta},\lambda)\left( 1+\mathcal{O}\left(\frac{1}{N^{1/2}|a\tau|}\right)\right),
    \end{split}
\end{equation}
with $\f$, $\g$, and $G$ defined in~\eqref{eq:explt}. The pre-factor $N^{3/2}$ in the leading term of~\eqref{eq:largenas} follows by a simple power counting: $a\sim N^{1/2}$, $\xi\sim N^{1/2}$, $\eta\sim N^{-1/2}$, the volume factor from the Jacobian of the change of variables~\eqref{eq:cv} gives a factor of $N$. In order to bound the error term in~\eqref{eq:leadterm} we also used the following lemma.

\begin{lemma}\label{lem:proteclem}
    Let $\f$ and $\g$ be the functions defined in~\eqref{eq:explt}. Then for any fixed  $\alpha, \beta,\gamma\in \R $ it holds
    \begin{equation}
        \label{eq:convdos}
        \int_{N^{-1/2}(\Gamma\setminus \widetilde{\Gamma})} \left|\dif \xi\right| \int_{\Omega\setminus \widetilde{\Omega}} \left|\dif \tau \right| \int_{N^{-1/2}(r\Lambda\setminus \widetilde{\Lambda})}  \left|\dif a\right| \left|\frac{1}{a^\alpha\tau^\gamma \xi^\beta} e^{\f(\xi,\lambda)-\g(a,\tau,\widetilde{\eta},\lambda)}\right|\le C,
    \end{equation}
    for some constant $C<\infty$ which depends only on $\alpha, \beta, \gamma$ and on
    the control parameters $C_0$, $C_1$ from Theorem~\ref{theo:1pointreal}.
\end{lemma}
\begin{proof}
    The bound in~\eqref{eq:convdos} directly follows from the explicit form of $\f$ and $\g$ in~\eqref{eq:explt} and by the fact that on the chosen contours $\Gamma$, $\Omega$, $r\Lambda$ it holds $\Re \g>0$, $\Re \f <0$.
\end{proof}

Using Lemma~\ref{lem:smallatau} once more, we can add back the regimes $\xi\in N^{-1/2}\widetilde{\Gamma}$, $\tau\in \widetilde{\Omega}$, $a\in N^{-1/2}\widetilde{\Lambda}$ to~\eqref{eq:leadterm}. Hence, using that we can deform the integration contours by holomorphicity, we conclude~\eqref{eq:largenas}--\eqref{eq:exp1poi}. The absolute convergence of $I^{(\R )}(\lambda,\widetilde{\eta})$ follows from Lemma~\ref{lem:proteclem}.

\subsection{The limit \texorpdfstring{$|\widetilde{\eta}|\to +\infty$}{|eta|->infty}.}\label{sec:tildlim}
The main goal of this section is to study the asymptotic of $I^{(\R )}(\lambda,\widetilde{\eta},\widetilde{\delta})$, defined in~\eqref{eq:exp1poi}, in the limit $|\widetilde{\eta}|\to +\infty$; in particular we prove that $I^{(\R )}(\lambda,\widetilde{\eta},\widetilde{\delta})$ converges to the $1$-point function of the shifted complex Ginibre ensemble $I^{(\C )}(\lambda,\widetilde{\delta})$, which is defined in~\eqref{eq:IC}. To make the presentation clearer, also in this case we present the proof only for the case $\widetilde{\delta}=0$ and denote $I^{(\R )}(\lambda,\widetilde{\eta}):=I^{(\R )}(\lambda,\widetilde{\eta},\widetilde{\delta}=0)$.

We recall that by Theorem~\ref{theo:1pointreal} we have
\begin{equation}
    \label{eq:ataunewform}
    I^{(\R )}(\lambda,\widetilde{\eta})=\frac{1}{4\pi \ii} \oint_{\Gamma} \dif \xi \int_\Omega \dif \tau \int_{\Lambda} \dif a \frac{\xi^2a}{\tau^{1/2}} e^{\f(\xi,\lambda)-\f(a,\lambda)} e^{\g(a,1,\widetilde{\eta},\lambda)-\g(a,\tau,\widetilde{\eta},\lambda)} G(a,\tau,\xi,\widetilde{\eta}),
\end{equation}
with $\Gamma$, $\Omega$, $\Lambda$ from Theorem~\ref{theo:1pointreal}, where we used that $\g(a,1,\widetilde{\eta},\lambda)=\f(a,\lambda)$ for any $a\in\C $.

\begin{proof}[Proof of Corollary~\ref{theo:1realcompl}]
    In this proof we use the notation
    \[
    \widetilde{\Omega}:=\{\tau\in\Omega : |\tau|\le C|\widetilde{\eta}|^{-1/2}\}, \quad \widetilde{\Lambda}:=\{a\in\Lambda : |a|\le |\widetilde{\eta}|^{-1/2}\},
    \]
    for some large constant $C>0$ (note that  $\widetilde{\Omega}$, $\widetilde{\Lambda}$ have already been used in~\eqref{eq:smallreg} to denote different segments). Then, similarly to the proof of Lemma~\ref{lem:smallatau}, it is easy to see that the integral in the regime when either $\tau\in\widetilde{\Omega}$ or $a\in\widetilde{\Lambda}$ is bounded by $e^{-c|\widetilde{\eta}|^{1/4}}$, for some small fixed $c>0$. In particular, by~\eqref{eq:ataunewform} we get that
    \begin{equation}
        \label{eq:interm}
        \begin{split}
            I^{(\R )}(\lambda,\widetilde{\eta})&=\frac{1}{4\pi \ii} \oint_{\Gamma} \dif \xi  \int_{\Lambda\setminus \widetilde{\Lambda}} \dif a\int_{\Omega\setminus \widetilde{\Omega}} \dif \tau \frac{\xi^2a}{\tau^{1/2}} e^{f(\xi,\lambda)-f(a,\lambda)} e^{\g(a,1,\widetilde{\eta},\lambda)-\g(a,\tau,\widetilde{\eta},\lambda)} G(a,\tau,\xi,\widetilde{\eta}) \\
            &\quad+\mathcal{O}\left(e^{-c|\widetilde{\eta}|^{1/4}}\right).
        \end{split}
    \end{equation}
    Note that by the definition of $G$ in~\eqref{eq:explt} the $\xi$-integral and the $(a,\tau)$-integral factorise, hence from now on we will consider only the $(a,\tau)$-integral.
    
    Then, to prove~\eqref{eq:largeetalim}, in the following lemma, whose proof is postponed to the end of this section, we compute the leading order term of the $\tau$-integral in~\eqref{eq:interm}. 
    
    \begin{lemma}\label{lem:sectauint}
        For any large constant $C_0>0$, and for any fix $\gamma\in \R $, $C_0^{-1}\le\lambda\le C_0$ it holds
        \begin{equation}
            \label{eq:tatata}
            \int_{\Omega\setminus\widetilde{\Omega}} \tau^{-\gamma} e^{\g(a,1,\widetilde{\eta},\lambda)-\g(a,\tau,\widetilde{\eta},\lambda)} \, \dif \tau=\frac{1}{2\widetilde{\eta}^2}+\mathcal{O}\left(|\widetilde{\eta}|^{-3}\right),
        \end{equation}
        uniformly in $a\in \Lambda\setminus \widetilde{\Lambda}$. The implicit constant in $\mathcal{O}(\cdot)$ depends on $C_0$.
    \end{lemma}

    Next, using~\eqref{eq:interm} and Lemma~\ref{lem:sectauint}, we conclude the proof of Corollary~\ref{theo:1realcompl}. First of all we notice that the leading term in~\eqref{eq:tatata} does not depend on $\gamma$, hence after performing the $\tau$-integration the power of $\tau$ that appears in $G(a,\tau,\xi,\widetilde{\eta})$ does not matter. For this reason after the $\tau$-integration we consider $G(a,1,\xi,\widetilde{\eta})$, i.e.\ for convenience we evaluate $G$ at $\tau=1$. More precisely, by Lemma~\ref{lem:sectauint} it follows that
    \begin{equation}
        \label{eq:intafttau}
        \int_{\Omega\setminus\widetilde{\Omega}}\frac{1}{\tau^{1/2}} e^{-\g(a,\tau,\widetilde{\eta},\lambda)}G(a,\tau,\xi,\widetilde{\eta})\, \dif \tau= e^{-\g(a,1,\widetilde{\eta},\lambda)}G(a,1,\xi,\widetilde{\eta})\left(\frac{1}{2\widetilde{\eta}^2}+\mathcal{O}\left(\frac{1}{|\widetilde{\eta}|^3}\right)\right).
    \end{equation}
    Then, by~\eqref{eq:interm} together with~\eqref{eq:intafttau}, it follows that
    \begin{equation}
        \label{eq:finetaco}
        \begin{split}
            I^{(\R )}(\lambda,\widetilde{\eta})&=\frac{1}{8\pi \ii} \oint_{\Gamma} \dif \xi  \int_{\Lambda\setminus\widetilde{\Lambda}} \dif a \frac{\xi^2a}{\widetilde{\eta}^2} e^{\f(\xi,\lambda)-\f(a,\lambda)} G(a,1,\xi,\widetilde{\eta})+\mathcal{O}\left(|\widetilde{\eta}|^{-1}\right) \\
            &= \frac{1}{2\pi \ii} \oint_{\Gamma} \dif \xi  \int_{\Lambda} \dif a  \left(\frac{1}{a\xi^2}+\frac{1}{a^2\xi}+\frac{1}{a^3}\right)e^{\f(\xi,\lambda)-\f(a,\lambda)}+\mathcal{O}\left(|\widetilde{\eta}|^{-1}\right),
        \end{split}
    \end{equation}
    where in the second equality we used the explicit form of $G$ from~\eqref{eq:explt}, and that we can add back the regime $a\in\widetilde{\Lambda}$ at the price of a negligible error. This concludes the proof of~\eqref{eq:largeetalim}.
\end{proof}

We now present the proof of Lemma~\ref{lem:sectauint}.

\begin{proof}[Proof of Lemma~\ref{lem:sectauint}]
    From now on we assume that $|\widetilde{\eta}|\ge C$, for some large constant $C>0$, since we are interested in the asymptotics for $|\widetilde{\eta}|\to +\infty$. Additionally, since
    \begin{equation}
        \label{eq:diffgs}
        \g(a,1,\widetilde{\eta},\lambda)-\g(a,\tau,\widetilde{\eta},\lambda)=-\frac{2\widetilde{\eta}^2(1-\tau)}{\tau}+\frac{\tau^2+\tau-2}{2a^2\tau^2}
    \end{equation}
    depends only on $\widetilde{\eta}^2$, without loss of generality we assume that $\widetilde{\eta}>0$.
    
    Next we split the $\tau$-integral in~\eqref{eq:tatata} into two parts: $|\tau|\in [0,1-\widetilde{\eta}^{-3/2})$ and $|\tau|\in [1-\widetilde{\eta}^{-3/2},1]$, which we denote by $\Omega_1$ and $\Omega_2$, respectively. It is easy to see that
    \begin{equation}
        \label{eq:tau012}
        \left|\int_{\Omega_1\setminus \widetilde{\Omega}} \tau^{-\gamma} e^{\g(a,1,\widetilde{\eta},\lambda)-\g(a,\tau,\widetilde{\eta},\lambda)}\, \dif \tau\right|\lesssim \int_{\Omega_1\setminus \widetilde{\Omega}} \frac{e^{-c\widetilde{\eta}^{1/2}\Re[\tau^{-1}]}}{|\tau|^{\gamma}}\, d\tau\lesssim e^{-c\widetilde{\eta}^{1/2}},
    \end{equation}
    for some small fixed $c>0$, where we used that by~\eqref{eq:diffgs} we have
    \[
    \Re[\g(a,1,\widetilde{\eta},\lambda)-\g(a,\tau,\widetilde{\eta},\lambda)]=-\Re\left[\frac{2\widetilde{\eta}^2(1-\tau)}{\tau}\left(1+\frac{2+\tau}{4a^2\tau\widetilde{\eta}^2}\right)\right]\le -c\Re\left[\frac{\widetilde{\eta}^{1/2}}{\tau}\right],
    \]
    for any $a\in \Lambda\setminus\widetilde{\Lambda}$ and $\tau\in \Omega_1\setminus \widetilde{\Omega}$. Hence, in order to conclude the proof, we are left only with the regime $\tau\in \Omega_2$ and $a\in \Lambda\setminus\widetilde{\Lambda}$.

    Define $t(\tau):= -e^{4\ii\pi/3}2\widetilde{\eta}^2(1-\tau)$, hence $\tau=\tau(t)=1+te^{-4\ii \pi/3}/(2\widetilde{\eta}^2)$, then we have that
    \[
    g(a,1,\widetilde{\eta},\lambda)-g(a,\tau(t),\widetilde{\eta},\lambda)=te^{-4\ii\pi/3}+\mathcal{O}\left(\frac{t}{\widetilde{\eta}^2}+\frac{t}{|a|^2\widetilde{\eta}^2}\right),
    \]
    and so that
    \begin{equation}
        \label{eq:tau121}
        \begin{split}
            \int_{\Omega_2} e^{g(a,1,\widetilde{\eta},\lambda)-g(a,\tau,\widetilde{\eta},\lambda)}\, \dif \tau &=\frac{e^{-4\ii\pi/3}}{2\widetilde{\eta}^2}\int_0^{t(\tau_0)} e^{t e^{-4\ii\pi/3}}\left[1+\mathcal{O}\left(\frac{t}{\widetilde{\eta}^2}+\frac{t}{|a|^2\widetilde{\eta}^2}\right)\right] \, \dif t \\
            &=\frac{1}{2\widetilde{\eta}^2}+\mathcal{O}\left(\frac{1}{|\widetilde{\eta}|^3}\right),
        \end{split}
    \end{equation}
    where $\tau_0:= \{|\tau|=1-\widetilde{\eta}^{-3/2}\}\cap \Omega$, and in the last equality we used that $|a|^{-2}\le \widetilde{\eta}$. Combining~\eqref{eq:tau012}-\eqref{eq:tau121} we conclude~\eqref{eq:tatata}.
\end{proof}

\appendix

\section{Additional technical results}\label{app:addlemm}
\begin{proof}[Proof of Lemma~\ref{lemma lower bound g}]
    The proofs of all lower bounds are similar, hence we will not provide details for all of them. For definiteness we will prove  cases~\ref{case1},~\ref{case2} and~\ref{case2b} since those already demonstrate the qualitatively different \(\abs{a\tau}\ll 1\), \(\abs{a\tau} \sim1\) and \(\abs{a\tau}\gg1\) regimes. 
    
    \begin{proof}[Proof of~\ref{case1}-\ref{case2}]
        First consider the \(\abs{a}\ll \abs{\tau}^{-1}\) case (note that this relation
        is necessarily fulfilled if \(\abs{\tau}\ll E\) 
        since then \(\abs{a}\lesssim E^{-1}\ll \abs{\tau}^{-1}\)).
        We have to prove \(\Re g\gtrsim -\log\abs{a\tau}\) which follows from
        \[\Re g \approx  - E\Re a+ \Re \Bigl[\frac{1}{2}\log \frac{1+2a}{a^2\tau}-\frac{1+a}{1+2a}\Bigr]
        \gtrsim  - E\Re a+ \frac{1}{2}\log \Big| \frac{1+2a}{a}\Big| + \frac{1}{2}\log\frac{1}{\abs{a\tau}}- 1\gtrsim -\log\abs{a\tau}.
        \]
        Thus we are only left with the \(\abs{a}^{-1}\sim\abs{\tau}\sim E\) case in which we introduce
        the parametrization \(\tau=t e^{-\ii\pi/3} E\) with  a real  \(t\sim 1\), so that 
        \[r=\frac{1\wedge t^{-2/3}}{2E} \Bigl(1+\landauO{E^{2/3}}\Bigr).\] 
        
        For \(\abs{a}\le r\) we have \(\arg(a)=\pi/6\) and parametrise \(a\tau=s e^{-\ii\pi/6}\) with 
        \begin{equation}\label{srange}
            s\in\Bigl[0,\frac{t\wedge t^{1/3}}{2}\frac{\sin(4\pi/15)}{\sin(17\pi/30)}\Bigr].
        \end{equation} 
        Using that $Ea = st^{-1} e^{\ii\pi/6}$ and that $|a|\gg 1$, the claim is thus equivalent to showing 
        \[
        \begin{split}
            &\Re g \approx -\frac{s}{t}\Re e^{\ii\pi/6}  + \Re\Bigl[ \frac{1}{2} \log\frac{2 + s e^{-\ii \pi/6} }{s e^{-\ii \pi/6}}
            -\frac{1}{2+s e^{-\ii \pi/6}}\Bigr]\\
            &= \frac{1}{52} \Bigl(-28 + 6 \sqrt{3} - \frac{26 \sqrt{3} s}{t} + 
            13 \log\Bigl( \frac{4}{s^2} + \frac{2\sqrt{3}}{s} + 1\Bigr)\Bigr) \gtrsim 1 + (\log s^{-1})_+
        \end{split} \]
        where the last inequality is valid for \(s\) as in~\eqref{srange} and any \(t\le 100\).
        
        We turn to the case \(\abs{a}>r\), i.e.\ to the second segment of the contour \(r\Lambda\) where we parametrise \[a=E^{-1}\biggl(\frac{1\wedge t^{-2/3}}{2} \frac{\sin(4\pi/15)}{\sin(17\pi/30)} e^{\ii\pi/6}+s e^{3\ii\pi/5}\biggr)=:E^{-1}\wt a.\]
        with \(s\in [0,\infty)\). We express \(\Re g\) in terms of $\wt a = \wt a(s)$ and differentiate it as a function of $s$ we see that that function
        \[s\mapsto \Re\biggl[ -\wt a - \frac{1}{2+\wt a t e^{-\ii\pi/3}} + \frac{1}{2} \log \Bigl( 1 + \frac{2e^{\ii\pi/3}}{\wt a  t}\Bigr)\biggr]\] has a local minimum of size \(\sim t^{-2/3}\) at \(s\sim t^{-2/3}\) and thus for \(t\lesssim 1\) we obtain 
        \(\Re g\gtrsim 1+\log(s^{-1})_+\) also in this final case, completing the proof. 
    \end{proof}
    \begin{proof}[Proof of~\ref{case2b}] 
        For \(\abs{a}\gg E^{-1}\sim r\) it follows that \(-E\Re a\sim E\abs{a}\) by the choice of contour $\Lambda$.
        Therefore it is sufficient to prove  
        \begin{equation}\label{large a claim}
            \Re g=  - E\Re a+  \Re \Bigl[\frac{1}{2}\log \frac{1+2a+a^2\tau}{a^2\tau}-\frac{1+a}{1+2a+a^2\tau}\Bigr] 
            \gtrsim  E\abs{a}.
        \end{equation}
        If \(\abs{a}\gtrsim \abs{\tau}^{-1}\) then we estimate  
        \[    
        \begin{split}
            \abs{\Re \Bigl[\frac{1}{2}\log \frac{1+2a+a^2\tau}{a^2\tau}-\frac{1+a}{1+2a+a^2\tau}\Bigr]}  &\approx  \abs{\Re \Bigl[\frac{1}{2}\log \frac{2+a\tau}{a\tau}-\frac{1}{2+a\tau}\Bigr]}  \\
            &\lesssim \abs{\Re\Big[\frac{1}{a^2\tau^2}\Bigr]}\sim \frac{1}{\abs{a\tau}^2}  \ll   E\abs{a},
        \end{split}\]
        where  in the first step we used that $|a|\gg 1$ and 
        in the last step we used \(\abs{a}\gg\abs{\tau}^{-1}\gtrsim \abs{\tau}^{-2/3}E^{-1/3}\), confirming~\eqref{large a claim}. On the other hand, if \(\abs{a}\ll\abs{\tau}^{-1}\) (but still $|a|\gg 1$) then 
        by Taylor expansion  in $|a\tau|\gg 1$ we obtain 
        \[\Re \Bigl[\frac{1}{2}\log \frac{1+2a+a^2\tau}{a^2\tau}-\frac{1+a}{1+2a+a^2\tau}\Bigr] \gtrsim \Re\frac{1}{2}\log \frac{1}{a\tau} >0,\]
        trivially confirming~\eqref{large a claim}. 
    \end{proof}
    \noindent The remaining cases~\ref{case3}--\ref{case12} may be estimated by similar elementary considerations. 
\end{proof}

\begin{proof}[Proof of Lemma~\ref{lem:dec}]
    The first assertion follows from elementary calculations resulting in 
    \begin{equation}
        \frac{\dif \Re g(x e^{\ii\pi/6},\tau,0,E)}{\dif x} \lesssim - \Re \frac{1}{x e^{\ii\pi/6}}<0.
    \end{equation}
    and the second assertion from 
    \begin{equation}
        \Re\Bigl[\frac{x^2 e^{\ii\pi/3} (1-\tau)}{1+2x e^{\ii\pi/6}+x^2 e^{\ii\pi/3}\tau}\Bigr]\ge 0
    \end{equation}
    using the definition of the \(\tau\)-contour \(\Omega\).
\end{proof}

\section{Lists of coefficients}
\subsection{Explicit coefficients for the real 1-point function integral representation}\label{app:expco}
Here we collect the explicit coefficients in~\eqref{expfxi2AAr}:
\[
\begin{split}
    c_{1,2,6,1}&=c_{1,6,2,3}=1, \quad c_{1,3,5,2}=c_{1,5,3,3}=2, \quad c_{1,4,4,3}=4, \\
    c_{2,2,5,1}&=c_{2,5,2,3}=2, \quad c_{2,3,4,2}=c_{2,4,3,3}=4, \\
    c_{3,2,4,1}&=c_{3,4,2,2}=1, \quad c_{3,3,3,2}=2, \\
    c_{4,2,4,1}&=c_{4,4,2,2}=1, \quad c_{4,3,3,2}=2.
\end{split}
\]

\subsection{Explicit formulas for the real symmetric integral representation}\label{appendix poly}
Here we collect the explicit formulas for the polynomials of \(a,\xi,\tau\) in the definition of \(G_N\) in~\eqref{eq:newbetG}.
\[ \begin{split}
    p_{2,0,0}&:= a^4 \tau ^2+2 a^3 \xi  \tau +4 a^3 \tau -a^2 \xi ^2 \tau +4 a^2 \xi ^2+8 a^2 \xi +2 a^2 \tau \\
    &\qquad +4 a^2+2 a \xi ^3+8 a \xi ^2+10 a \xi +4 a+\xi ^4+4 \xi ^3+6 \xi ^2+4 \xi +1,   \\
    p_{1,0,0}&:= - a^4 \xi \tau ^2+a^4 \tau ^2-2 a^3 \xi ^2 \tau -2 a^3 \xi  \tau +4 a^3 \tau -a^2 \xi ^3 \tau -3 a^2 \xi ^2 \tau \\
    &\qquad -2 a^2 \xi  \tau +4 a^2 \xi +2 a^2 \tau +4 a^2+2 a \xi ^2+6 a \xi +4 a+\xi ^3+3 \xi ^2+3 \xi +1,\\
    p_{2,2,0}&:= 4 (a+1) \left(a^2 \tau +a \xi  \tau +2 a \tau +\xi ^2+2 \xi +1\right), \\
    p_{1,2,0}&:= 4 (a+1)  \left(a^2 \tau +a \xi  \tau +2 a \tau +\xi +1\right), \\
    p_{2,0,1}&:= 2  \bigl(a^3 \tau ^2+2 a^2 \xi  \tau +4 a^2 \tau +2 a \xi ^2+2 a \xi  \tau  \\
    &\qquad\qquad+4 a \xi +3 a \tau +2 a+\xi ^3+4 \xi ^2+5 \xi +2\bigr)\\
    p_{1,0,1}&:= 2 \bigl(a^3 \tau ^2+2 a^2 \xi  \tau +4 a^2 \tau +a \xi ^2 \tau +3 a \xi  \tau \\
    &\qquad\qquad+2 a \xi +3 a \tau  +2 a+\xi ^2+3 \xi +2\bigr), \\
    p_{2,2,1}&:= 4 (a+1)  (a+\xi +2), \\
    p_{2,0,2}&:= a^2 \tau +2 a \xi +4 a+\xi ^2+4 \xi +4. 
\end{split} \]

\printbibliography%
\end{document}

%% file: notations.tex
\DeclareMathOperator{\Tr}{Tr}

\DeclareMathOperator{\E}{\mathbf{E}}
\DeclareMathOperator{\Prob}{\mathbf{P}}

\newcommand{\ii}{\mathrm{i}}

\ifdefined\C
\renewcommand{\C}{\mathbf{C}}
\else
\newcommand{\C}{\mathbf{C}}
\fi

\newcommand{\wt}{\widetilde}

\newcommand{\R}{\mathbf{R}}
\newcommand{\N}{\mathbf{N}}
\newcommand{\Z}{\mathbf{Z}}

\newcommand{\DD}{\mathbf{D}}

\newcommand{\cO}{\mathcal{O}}
\newcommand{\co}{{\scriptstyle\mathcal{O}}}

\newcommand{\dif}{\operatorname{d}\!{}}
\DeclarePairedDelimiter{\abs}{\lvert}{\rvert}%
\providecommand\given{}
\newcommand\SetSymbol[1][]{\nonscript\:#1\vert\allowbreak\nonscript\:\mathopen{}}
\DeclarePairedDelimiterX{\tuple}[1](){\renewcommand\given{\SetSymbol[\delimsize]}#1}
\DeclarePairedDelimiterX{\set}[1]{\{}{\}}{\renewcommand\given{\SetSymbol[\delimsize]}#1}
\DeclarePairedDelimiterX{\Set}[1]\{\}{\renewcommand\given{\SetSymbol[\delimsize]}#1}
\DeclarePairedDelimiterXPP{\landauO}[1]{\cO}(){}{#1}
\DeclarePairedDelimiterXPP{\landauo}[1]{\co}(){}{#1}
\DeclarePairedDelimiterXPP{\landauok}[1]{\co_k}(){}{#1}
\DeclarePairedDelimiterXPP{\landauOprec}[1]{\cO_\prec}(){}{#1}
\DeclarePairedDelimiterXPP{\Exp}[1]{\E}[]{}{\renewcommand\given{\SetSymbol[\delimsize]}#1}
\DeclarePairedDelimiterXPP{\condProb}[1]{\Prob}[]{}{\renewcommand\given{\SetSymbol[\delimsize]}#1}
\DeclareFontFamily{U}{mathx}{\hyphenchar\font45}
\DeclareFontShape{U}{mathx}{m}{n}{
      <5> <6> <7> <8> <9> <10>
      <10.95> <12> <14.4> <17.28> <20.74> <24.88>
      mathx10
      }{}
\DeclareSymbolFont{mathx}{U}{mathx}{m}{n}
\DeclareFontSubstitution{U}{mathx}{m}{n}
\DeclareMathAccent{\widecheck}{0}{mathx}{"71}

%% file: bibliography_setup.tex
\usepackage[giveninits=true,url=false,doi=false,isbn=false,eprint=true,datamodel=mrnumber,sorting=nty,maxcitenames=4,maxbibnames=99,backref=false,block=space,backend=biber,bibstyle=phys]{biblatex} 
\AtEveryBibitem{\clearfield{month}}
\AtEveryCitekey{\clearfield{month}} 
\renewbibmacro{in:}{}
\DeclareFieldFormat[article]{title}{\emph{#1}} 
\DeclareFieldFormat{mrnumber}{\ifhyperref{\href{http://www.ams.org/mathscinet-getitem?mr=#1}{\nolinkurl{MR#1}}}{\nolinkurl{#1}}}
\DeclareFieldFormat{pmid}{\ifhyperref{\href{https://www.ncbi.nlm.nih.gov/pubmed/#1}{\nolinkurl{PMID#1}}}{\nolinkurl{#1}}}
\DeclareFieldFormat{eprint}{\ifhyperref{\href{https://arxiv.org/abs/#1}{\nolinkurl{arXiv:#1}}}{\nolinkurl{#1}}}
\renewbibmacro*{doi+eprint+url}{%
  \iftoggle{bbx:doi}{\printfield{doi}}{}
  \newunit\newblock%
  \printfield{mrnumber}%
  \newunit\newblock%
  \printfield{pmid}%
  \newunit\newblock%
  \printfield{eprint}%
  \iftoggle{bbx:url}{\usebibmacro{url+urldate}}{}}